\documentclass[a4paper,10pt]{article}
\usepackage[T1]{fontenc}
\usepackage[utf8]{inputenc}
\usepackage{t1enc,a4,times}
\usepackage{verbatim,moreverb,alltt}
\usepackage[english]{babel}
\usepackage[dvips]{graphics}
\usepackage{epsfig}
\usepackage{fullpage}
\usepackage{amsfonts}
\usepackage{multicol}

\usepackage{latexsym,amssymb}     % Pour certains symboles (\Box, ...)
\usepackage{amsmath}
\usepackage{indentfirst}
\usepackage{makeidx}
\usepackage{theorem}
\usepackage{bbm}
\usepackage{mathrsfs}
\usepackage{psfrag}
\usepackage{enumitem}
\usepackage{epsfig}
\usepackage{tikz}
\usepackage[ruled,lined]{algorithm2e}
\usepackage{multirow}

\newcommand{\cL}{\mathcal{L}}

\newcommand{\R}{\mathbb{R}}

\newcommand{\E}{\mathbb{E}}

\newcommand{\Px}{\mathbb{P}}

\newcommand{\indi}[1]{\mathbbm{1}_{#1}}

\newcommand{\varmaj}{H}
\newcommand{\varmin}{h}

\newcommand{\blanc}[1]{\vspace{#1\baselineskip}}

\newtheorem{nt_theorem}{Theorem}%[section]

\newtheorem{nt_proposition}[nt_theorem]{Proposition}
\newenvironment{proposition}{\blanc{0.05}\begin{nt_proposition}---}{\end{nt_proposition}\blanc{0.05}}

\newtheorem{nt_corollaire}[nt_theorem]{Corollary}

\newtheorem{nt_definition}[nt_theorem]{Definition}
\newenvironment{definition}{\blanc{0.05}\begin{nt_definition}---}{\end{nt_definition}\blanc{0.05}}

\newtheorem{nt_lemma}[nt_theorem]{Lemma}
\newenvironment{lemma}{\blanc{0.05}\begin{nt_lemma}---}{\end{nt_lemma}\blanc{0.05}}

\newtheorem{nt_remark}[nt_theorem]{Remark}
\newenvironment{remark}{\blanc{0.05}\begin{nt_remark}---}{\end{nt_remark}\blanc{0.05}}

\newtheorem{nt_exemple}[nt_theorem]{Example}

\newenvironment{proof}{{\textit{Proof : }}}{\hfill$\Box$
\\\bigskip}

\makeatletter
\newcounter{hypo}
\newcommand*{\symm}{{\mathcal{S}_d(\mathbb R)}}

\newcommand*{\genm}{{\mathcal{M}_d(\mathbb R)}}

\newcommand*{\posm}{{\mathcal{S}_d^+(\mathbb R)}}

\newcommand{\fctgm}{g}

\newcommand*{\dohypo}{\textbf{(${\mathcal H}$\thehypo)}}
\newcommand*{\Tr}{\textup{Tr}} %trace
\def\hypref#1{\hyperref[hyp:#1]{(${\mathcal H}$\ref*{hyp:#1})}}
\def\hypreff#1#2{\hyperref[hyp:#2]{(${\mathcal H}$\ref*{hyp:#1}-{\it
\ref*{hyp:#2})}}}

\title{Smile with the Gaussian term structure model}%
\author{Abdelkoddousse Ahdida, Aur\'{e}lien Alfonsi\footnotemark[1], Ernesto Palidda\thanks{Universit\'e Paris-Est, CERMICS, Projet MATHRISK
    ENPC-INRIA-UMLV, 6-8 avenue Blaise Pascal, 77455 Marne La Vall\'ee, France. \newline E-mails : ahdida.abdel@gmail.com, alfonsi@cermics.enpc.fr, ernesto.palidda@gmail.com. \newline
Most of this work has been done when Ernesto was working for the Groupe de Recherche Op\'erationnelle of Cr\'edit Agricole. We thank Christophe Michel, Victor Reutenauer, Anas Benabid and Nicole El Karoui for interesting and fruitful discussions around this paper.  This research also benefited from the support of the ``Chaire Risques Financiers'', Fondation du Risque.
}}%
\date{\today}%
\begin{document}
\maketitle

\vspace{-0.4cm}
\begin{abstract} We propose an affine extension of the Linear Gaussian term structure Model (LGM) such that the instantaneous covariation of the factors is given by an affine process on semidefinite positive matrices. First, we set up the model and present some important properties concerning the Laplace transform of the factors and the ergodicity of the model. Then, we present two main numerical tools to implement the model in practice. First, we obtain an expansion of caplets and swaptions prices around the LGM. Such a fast and accurate approximation is useful for assessing the model behavior on the implied volatility smile. Second, we provide a second order scheme for the weak error, which enables to calculate exotic options by a Monte-Carlo algorithm. These two pricing methods are compared with the standard one based on Fourier inversion. 
\end{abstract}


\section*{Motivation and overview of the paper}

Affine Term Structure Models (ATSM) are an important class of models for interest rates that include the classical and pioneering models of Vasicek~\cite{Vasicek} and Cox-Ingersoll-Ross~\cite{CIR1}. These models have been settled and popularized by the papers of Duffie and Kan~\cite{DK1996}, Dai and Singleton~\cite{DS1997} and Duffie, Filipovi\'c and Schachermayer~\cite{DFS}. We refer to Filipovi\'c~\cite{Filipovic} for a textbook on these term structure models. The Linear Gaussian Model (LGM) is a simple but important subclass of ATSM that assumes that the underlying factors follow a Gaussian process. It has been considered by  El Karoui and Lacoste~\cite{NEK1991} and El Karoui et al.~\cite{NEK1992}, and has now become a market standard for pricing fixed income derivatives, thanks to its simplicity. However, this model has a main drawback to be calibrated to market data: it produces implied volatility smiles that are flat. %

The goal of this paper is to present a quite natural extension of the LGM that keeps the affine structure and generates an implied volatility smile. To do so, we consider an affine diffusion  of Wishart type on the set of semidefinite positive matrices and replace, roughly speaking, the constant volatility matrix by (a linear function of) this process. The dependence between the factors and their volatility is made through a specific covariation that keeps the affine structure and that has been proposed by Da Fonseca et al.~\cite{F} in an equity framework. Thanks to this, the proposed model which is a stochastic variance-covariance affine term structure model (see Definition~\ref{ModelDef}), is able to produce an implied volatility smile. It  has many parameters and may seem at first sight difficult to handle. For this reason, we present it as a perturbation of the LGM. Thus, the calibration of the model to market data can be made in two steps: first, one can calibrate the LGM and then calibrate the new parameters to the implied volatility smile. The calibration of this model is discussed on some cases in Palidda~\cite{Palidda}. In the present paper, we do not tackle the practical calibration issue: our goal is just to set up the model and give the main numerical methods for a practical use of this model. Namely, we define in Section~\ref{Sec_defmodel} the model and present some important properties such as the value of the Laplace transform under the initial and forward measures or the ergodicity property. Then, we give two tools that are important to implement the model in practice. First, we present in Section~\ref{Sec_PriceExp} a price expansion for caplets and swaptions around the LGM when  the volatility of the volatility of the factor~$Y$ is small. These explicit formulas are useful to calculate quickly the impact of the parameters on the volatility cube and thus to calibrate the model. Second, we propose in Section~\ref{Sec_scheme} a discretization scheme for the model that is of second order for the weak error. Having an accurate scheme is important in practice since it allows to calculate exotic options by a Monte-Carlo algorithm. Besides, this scheme can be easily adapted to other models relying on the same affine structure such as the one of Da Fonseca et al.~\cite{F}. Last, Section~\ref{sec_comp_meth} compares the expansion and the Monte-Carlo method with the classical Fourier technique popularized by Carr and Madan~\cite{CM} and indicates the relevance of each method.

\section{The Linear Gaussian Model (LGM) in a nutshell}\label{Preliminaries}

The model that we present is meant to extend the classical LGM, and we need thus to recall briefly the LGM. We work under a risk-neutral measure~$\mathbb{P}$, and consider a $p$-dimensional standard Brownian motion~$Z$. Let $Y$  be the solution of the following Ornstein-Uhlenbeck SDE
\begin{equation}\label{LGMdynamics}
    Y_t=y+\int_0^t\kappa(\theta -Y_s) ds +\int_0^t\sqrt{V}dZ_s,
\end{equation}
\noindent where   $\kappa \in \mathcal M_p\left(\R\right)$ is a matrix of order~$p$,  $V$ is a semidefinite positive matrix of order~$p$ and $\theta \in \R^p$.
The LGM assumes that the spot rate is an affine function of the vector $Y$:
\begin{equation}\label{SpotRateLGM}
    r_t=\varphi+\sum_{i=1}^p Y^i_t,
\end{equation}
and the coordinates $Y^i$ are usually called the factors of the model. It is not restrictive to assume that the weight of each factor in~\eqref{SpotRateLGM} is the same for all factors and equal to one: if we had $r_t=\varphi+\sum_{i=1}^p m_i Y^i_t$, we could check easily that $(m_1Y^1,\dots,m_pY^p)^{\top}$ is also an Ornstein-Uhlenbeck process. Affine Term Structure models generally assume that the parameters (here $\kappa$, $\theta$ and $V$) are fixed and are valid over a long time period, while the factors (here the vector~$Y$) evolve and reflect the current state of the market. Therefore, one often assumes that the process~$Y$ is stationary to reflect some market equilibrium. Also, the factors are usually associated to different time scales: a factor with a small (resp. large) mean-reversion will influence the long-term (resp. short-term) behaviour of the interest rate.
This leads to assume that $$\kappa=\mathrm{diag}(\kappa_1,\dots,\kappa_p) \text{ with } 0<\kappa_1<\dots<\kappa_p, $$
 and we work under this assumption in the sequel. It can be easily checked (see for example Andersen and Piterbarg~\cite{Piterbarg2010}) that any linear Gaussian model such that $\kappa$ has distinct positive eigenvalues can be rewritten, up to a linear transformation of the factors, within the present parametrization.

Let $(\mathcal{F}_t)_{t\ge 0}$  denote the natural filtration of~$Z$. For $0\le t\le T$, the price $P_{t,T}= \E\left[ \exp\left(-\int_t^T r_s ds \right)|\mathcal{F}_t \right]$  at time~$t$ of the zero-coupon bond with maturity~$T$ is an exponential affine function of~$Y$:
\begin{equation}\label{ZCLGM}
    P_{t,T}=\exp(E(T-t)+B(T-t)^{\top}Y_t),
\end{equation}
where $B(\tau) =-(\kappa^{\top})^{-1}(I_p-e^{-\kappa^{\top}\tau}) \mathbf{1}_p$ and $E(\tau) = -\varphi\tau+\int_0^\tau B(s)^{\top}\kappa\theta+\frac{B(s)^{\top}VB(s)}{2} ds$ for $\tau \ge 0$. Here, $\mathbf{1}_p$ stands for the vector in $\R^p$ that has all its entries equal to one.  The function $B(\tau)$ maps the factors variations $\Delta Y$ on the yield curve variations and is often called the support function.  The factors $Y^i$ associated with the larger parameters $\kappa_i$ impact on the short term behaviour of the yield curve  while the one  associated with the smaller parameters $\kappa_i$ will drive more the long term behaviour.

\noindent We now briefly introduce some of the basic notions on the interest rates vanilla option market. The most liquid traded interest rates options are swaptions and caplets. They are respectively expressed with respect to the forward Libor rate and the forward swap rate, which are defined as follows for $0\le t \le T$, $\delta>0$ and $m \in \mathbb{N}^*$:
\begin{eqnarray*}
%  \label{forwardLibor}
 L_t(T,\delta) &=& \frac{1}{\delta}\left(\frac{P_{t,T}}{P_{t,T+\delta}}-1\right) \\
% \label{forwardSwap}
 S_t(T,m) &=& \frac{P_{t,T}-P_{t,T+m\delta}}{\delta \sum_{i=1}^mP_{t,T+i\delta}}.
\end{eqnarray*}

\noindent The prices of caplets and swaptions are respectively given by

\begin{eqnarray*}
  C_t(T,\delta,K) &=& \E\left[e^{-\int_t^{T+\delta}r_sds}\left(L_{T}(T,\delta)-K\right)^+ \bigg|\mathcal F_t\right] \\
  \mathrm{Swaption}_t(T,m,\delta,K) &=& \E\left[e^{-\int_t^{T}r_s ds}\sum_{i=1}^m \delta P_{T,T+i\delta}\left(S_{T}(T,m)-K\right)^+ \bigg|\mathcal F_t\right].
\end{eqnarray*}
Caplets are usually available for short tenors~$\delta$ (up to 1 year) and swaptions are quoted for tenors~$m\delta$ from 2 to 30 years. The market practice is to apply a standard change of numeraire technique (see Geman et al.~\cite{NEK95}) and rewrite the above expressions as

\begin{eqnarray}
  \label{CapletPricing} C_t(T,\delta,K) &=& P_{t,T+\delta}\E^{T+\delta}\left[\left(L_{T}(T,\delta)-K\right)^+|\mathcal F_t\right] \\
  \label{SwaptionPricing} \mathrm{Swaption}_t(T,m,\delta,K) &=& \left( \sum_{i=1}^m\delta P_{t,T+i\delta} \right) \E^{A}\left[\left(S_{T}(T,m) -K\right)^+|\mathcal F_t\right],
\end{eqnarray}

\noindent where $\E^{T+\delta}$ (resp. $\E^{A}$) denotes the expectation taken with respect to the measure $T+\delta$-forward neutral (resp. annuity) measure associated with the numeraire $P_{t,T+\delta}$ (resp. $\sum_{i=1}^m \delta P_{t,T+i\delta}$).   The market prices are then quoted and analyzed in terms of either the log-normal or normal implied volatility obtained by inverting respectively the pricing formulas (\ref{CapletPricing}) and (\ref{SwaptionPricing}) w.r.t. the Black-Scholes and Bachelier formulas. Within the LGM model, the log-normal implied volatility of the caplet is given by
$$\int_t^{T} [B(T-u)-B(T+\delta-u)]^{\top} V  [B(T-u)-B(T+\delta-u)] du, $$
which is a particular case of formula~\eqref{vcaplet} below. This implied volatility does not depend on the strike. It shows that the mean-reversion parameter $\kappa$ plays a role in shaping the form of the caplets volatility cube, according to the different time scales. The role of the diagonal coefficients of the matrix $V$ is determined by the support functions $m_{ii}(\tau,\delta)=(\frac{1-e^{-\kappa_i\delta}}{\kappa_i})^2\frac{1-e^{-2\kappa_i\tau}}{2\kappa_i\tau}$.  The effect of off-diagonal elements of the variance-covariance matrix $V$ is determined by the support functions $m_{ij}(\tau,\delta)=\frac{1-e^{-\kappa_i\delta}}{\kappa_i}\frac{1-e^{-\kappa_j\delta}}{\kappa_j}\frac{1-e^{-(\kappa_i+\kappa_j)\tau}}{(\kappa_i+\kappa_j)\tau}$. These functions are plotted in Figure~\ref{SupportFunctionVol}.

Also, by using a standard approximation, we can obtain the normal implied volatility of the swaptions:
$$ \int_t^{T} [B^S(u)]^{\top} V  B^S(u) du,$$
with $B^S(u)=\omega^0_0B(T-u)-\omega^m_0B(T+m\delta-u)-S_0(T,m,\delta)\sum_{k=1}^m\omega^k_0B(T+k\delta-u)$, $\omega^k_0 = \frac{P_{0,T+k \delta}}{\sum_{i=1}^m P_{0,T+i \delta}}$. This is a particular case of formula~\eqref{VolSwaption} below. This implied volatility has a rather similar structure as the one of the caplets, but it is not time homogeneous.
\begin{figure}[h!]
\centering
\includegraphics[scale = 0.28]{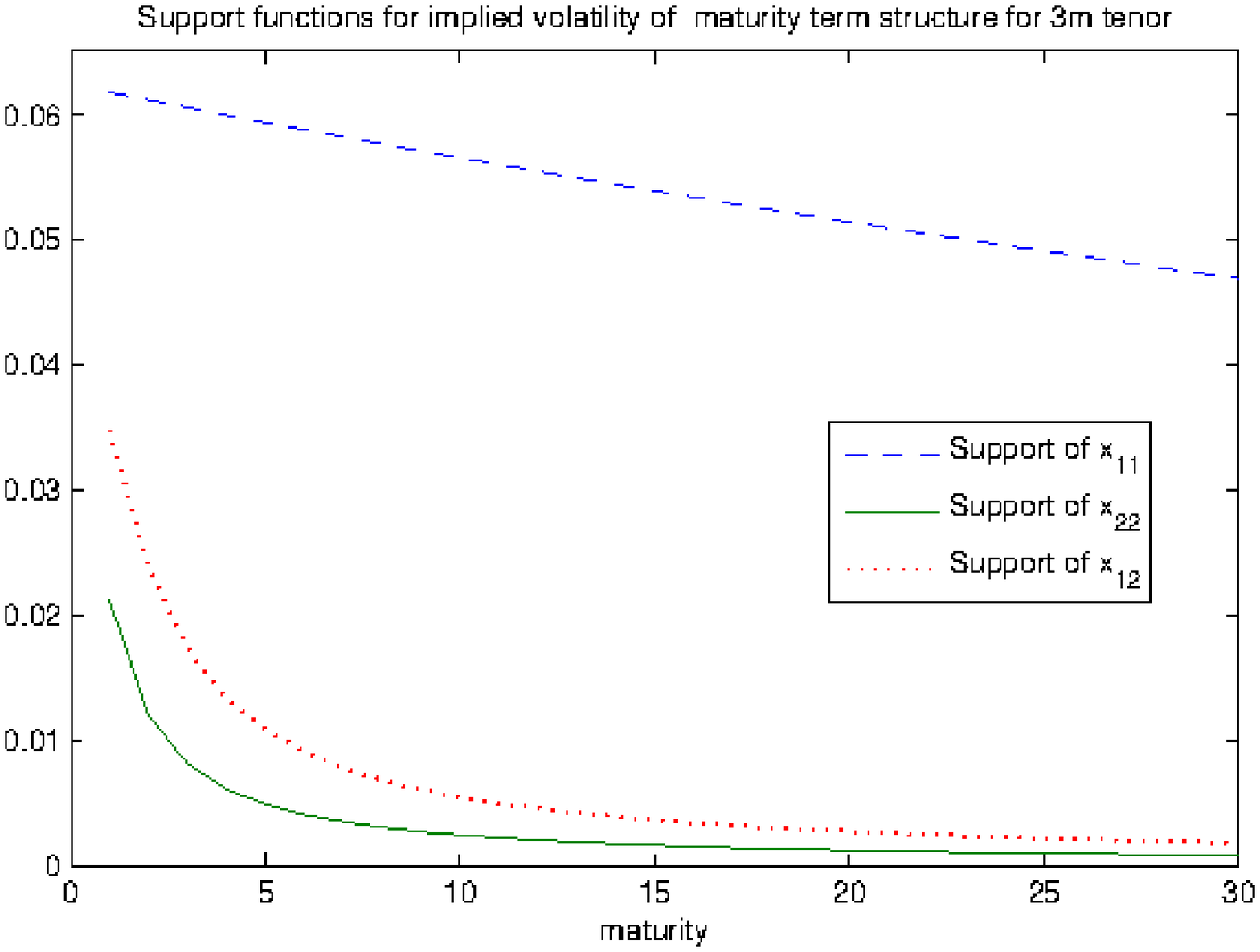}
\includegraphics[scale = 0.28]{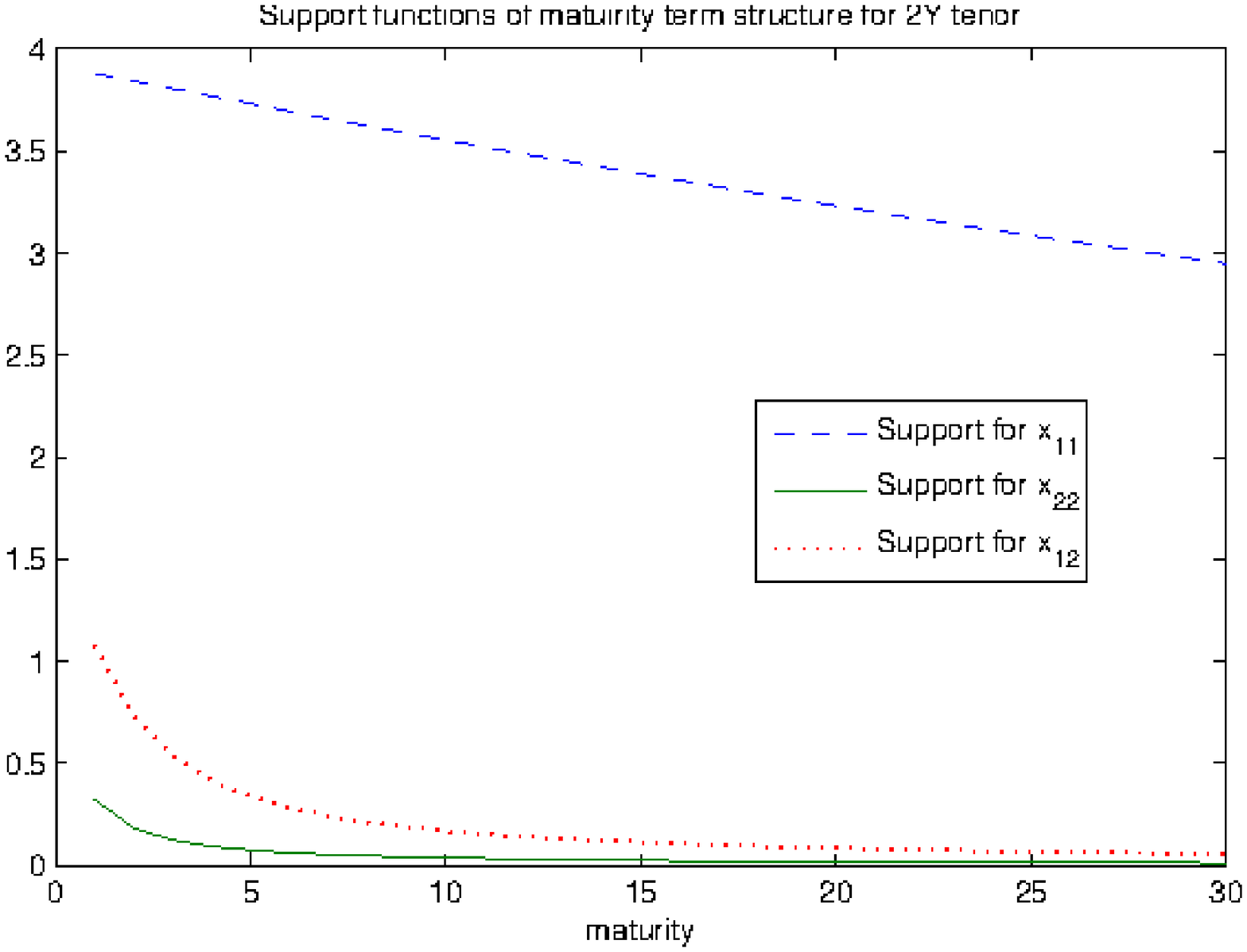}\\
\caption{Support functions for the volatility term structure in a two factors model with $\kappa=\mathrm{diag}(0.01,1)$ for a 3 months (left) and 2 years (right) maturity.}
\label{SupportFunctionVol}
\end{figure}
 Both implied volatilities for caplets and swaptions do not depend on the strike and give thus a flat smile, which is a well-known fact. This is unfortunate if one aims to reproduce the volatility cube observed on market data (i.e. the implied volatility with respect to the maturity~$T$, the tenor $\delta$ or $m\delta$ and the strike $K$). The extension of the LGM that we introduce in Section~\ref{Sec_defmodel} is meant to correct this drawback.

%***********************************************************************************************************************************************************
%***********************************************************************************************************************************************************

\section{An affine extension of the LGM with stochastic covariance}\label{Sec_defmodel}

This section is devoted to the definition of the model that we study in this paper. This model is a stochastic variance-covariance perturbation of the LGM. We chose a quite general specification that keeps the model affine and gives a stochastic instantaneous covariance for the factors, which will generate a smile for the Caplets and Swaptions. We first present the dynamics of the factor and then present some properties of the model that rely on its affine structure.

\subsection{State variables dynamics}\label{StateVarSec}

We consider $W$ a $d$-by-$d$ square matrix made of independent standard Brownian motions, and $Z$ an independent Brownian motion of dimension~$p$. We will denote by $(\mathcal{F}_t)_{t\ge 0}$ the filtration generated by $(W,Z)$. We consider the following SDE for the state variables (or factors)
\begin{eqnarray} \label{SDE_GeneralY}
 Y_t &=&y +\int_0^t\kappa(\theta -Y_s) ds + \int_0^tc\sqrt{X_s}\left[\bar\rho dZ_s + dW_s\rho \right]\\
X_t &=&x +\int_0^t\left( \Omega +(d-1)\epsilon^2I^n_d  + bX_s+X_sb^{\top} ) \right )ds + \epsilon\int_0^t\sqrt{X_s}dW_sI^n_d + I^n_ddW_s^{\top} \sqrt{X_s}. \label{SDE_General}
\end{eqnarray}
\noindent The matrix $I^n_d$ is defined for $0\le n\le d$ by $(I^n_d)_{i,j}=\indi{i=j\le n}$, and the parameters are taken as follows
\begin{eqnarray}\label{SDE_General_Params}
   & &\nonumber x, \Omega \in \posm, b\in \genm, \epsilon\in \R_+, \ y,\theta \in \R^p, \kappa=\mathrm{diag}(\kappa_1,...,\kappa_p) \text{ with } \kappa_1,...,\kappa_p>0,\\
   & & c\in \mathcal{M}_{p\times d}(\R), \rho\in \R^d \text{ such that } |\rho|^2:=\sum_{i=1}^d \rho_i^2 \leq 1 \quad\mathrm{and}\quad \bar\rho=\sqrt{1-|\rho|^2},
\end{eqnarray}
where $\posm$, $\genm$, and $\mathcal{M}_{p\times d}(\R)$ denote respectively the set of semidefinite positive matrices of order $d$, the set of square matrices of order~$d$, and the set of matrices with $p$ rows and $d$ columns. The process $X$ is an affine diffusion on~$\posm$, and the instantaneous covariance at time~$t$ of the factors $Y$ is given by $cX_tc^\top$. When $\epsilon=0$ and $\Omega=-bx-xb^{\top}$, we have $X_t=x$ and get back the Gaussian model with $V=cxc^\top$. The dependence structure between $Y$ and $X$ through the driving Brownian motions is the same as the one proposed by Da Fonseca, Grasselli and Tebaldi~\cite{F}. As explained in~\cite{F}, this is the most general way to get a non trivial instantaneous correlation between $Y$ and $X$ while keeping the affine structure. In particular, the instantaneous quadratic covariations are linear with respect to~$(Y,X)$ and we have for $1\le i,j,k,l\le d$ and $1\le m,m'\le p$:
\begin{align}\label{crochets_1}
\langle d(Y_t)_m,d(Y_t)_{m'} \rangle&=(cX_tc^\top)_{m,m'}dt, \\
  \langle d(X_t)_{i,j},d(X_t)_{k,l} \rangle&=\epsilon^2 \left[(X_t)_{i,k}\indi{j=l\le n}+(X_t)_{i,l}\indi{j=k\le n} + (X_t)_{j,k}\indi{i=l\le n}+  (X_t)_{j,l}\indi{i=k\le n} \right]dt ,\label{crochets_2} \\
\langle d(Y_t)_m,d(X_t)_{i,j} \rangle&=\epsilon \left[ (cX_t)_{m,i} (I^n_d \rho)_j + (cX_t)_{m,j} (I^n_d \rho)_i\right] dt.\label{crochets_3}
\end{align}
We notice that only the $n$ first components of~$\rho$ matter, and we can assume without loss of generality that $\rho_{n+1}=\dots=\rho_d=0$.

From~\eqref{SDE_General}, we easily get
$$e^{\kappa t}Y_t=y + \int_0^t e^{\kappa s} \kappa \theta ds  + \int_0^t e^{\kappa s} c\sqrt{X_s}\left[\bar\rho dZ_s + dW_s\rho \right]. $$
Therefore, the process $Y$ is uniquely determined once the processes $Z$, $W$ and $X$ are given. We know by Cuchiero et al.~\cite{CFMT2011} that the SDE on~$X$ has a unique weak solution when $x \in\posm$ and $\Omega \in \posm$, and a unique strong solution if we assume besides that $x$ is invertible and  $\Omega- 2\epsilon^2I^n_d \in \posm$. This leads to the following result.

\begin{proposition}\label{existence}
If $x\in\posm$, $\Omega\in \posm$  there exists a unique weak solution of the SDE~\eqref{SDE_General}. If we assume moreover that $\Omega- 2\epsilon^2I^n_d \in \posm$ and $x\in\posm$ is positive definite, there is a unique strong solution to the SDE (\ref{SDE_General}).
\end{proposition}

\noindent The affine structure of the process $(X,Y)$ allows us to give formulas for the Laplace transform of the marginal laws by means of  Matrix Riccati Differential Equations (MRDE). Similar calculations have been made in equity modelling by Da Fonseca et al.~\cite{F} or Benabid et al.~\cite{Benabid}.    The following proposition states the precise result, which is useful for the pricing of Zero-Coupon bonds.

\begin{proposition}\label{FourLapGeneralProp}
Let $\Lambda,\bar{\Lambda} \in \R^p$, $\Gamma, \bar \Gamma \in \symm$. For $t\ge 0$, we define $\lambda(t) \in \R^p$ by
\begin{equation}\label{val_lambda}\lambda_i(t)=\Lambda_i e^{- \kappa_i t} + \frac{  \bar\Lambda_i}{\kappa_i} (1- e^{- \kappa_i t} ).
\end{equation}
Let us assume that there exists $\Upsilon \in \symm$ such that
\begin{align}
&\Upsilon-\Gamma \in \posm, \label{cond1_ups} \\ &\forall t\ge 0,- \left[ 2\epsilon^2 \Upsilon I^n_d\Upsilon + \Upsilon(b+ \frac{\epsilon}{2} I^n_d\rho \lambda^{\top} c) +(b+\frac{\epsilon}{2} I^n_d\rho \lambda^{\top} c)^{\top} \Upsilon + \frac{1}{2}c^{\top} \lambda\lambda^{\top} c+\bar\Gamma \right] \in \posm \label{cond2_ups}
\end{align}
Then, the following system of differential equations
\begin{equation}
\label{LaplaceGenRiccati}
\left\{
\begin{array}{l}
\dot\fctgm=2\epsilon^2\fctgm I^n_d\fctgm+\fctgm (b+ \frac{\epsilon}{2} I^n_d\rho \lambda^{\top} c)+(b+\frac{\epsilon}{2} I^n_d\rho \lambda^{\top} c)^{\top} \fctgm+\frac{1}{2}c^{\top} \lambda\lambda^{\top} c+\bar\Gamma, \ \fctgm(0)=\Gamma,\\
\\
\dot\eta=\lambda^{\top} \kappa\theta+\mathrm{Tr}\left(\fctgm(\Omega+\epsilon^2(d-1)I^n_d)\right),\quad \eta(0)=0,
\end{array}
\right.
\end{equation}
has a unique solution, which is defined on $\R_+$. It satisfies $\Upsilon -g(t) \in \posm$ for any $t\ge 0$.  Besides, we have for any $0\le t \le T$:
\begin{equation}\label{laplace}
    \E \left[\exp\left(\Tr(\Gamma X_T)+\Lambda^{\top} Y_T+\int_t^T  \Tr\left(\bar\Gamma X_s\right)+\bar\Lambda^{\top} Y_sds \right)\bigg|\mathcal{F}_t\right]=\exp(\eta(T-t)+\Tr(\fctgm(T-t)X_t)+\lambda(T-t)^{\top} Y_t).
\end{equation}
\end{proposition}
\begin{proof}
The proof is quite standard for affine diffusion. First, we notice that if~\eqref{laplace} holds, we necessarily have that $M_t= \exp\left( \int_0^t  \Tr\left(\bar\Gamma X_s\right)+\bar\Lambda^{\top} Y_s ds \right) \exp(\eta(T-t)+\Tr(\fctgm(T-t)X_t)+\lambda(T-t)^{\top} Y_t)$ is a martingale. We apply It\^o's formula and use~\eqref{crochets_1},~\eqref{crochets_2} and~\eqref{crochets_3}. The martingale property yields to
\begin{align*}
&\bar\Gamma X_t+\bar\Lambda^\top Y_t-\dot\eta(T-t)-\Tr(\dot\fctgm(T-t)X_t)-\dot\lambda(T-t)^{\top} Y_t + \Tr(\fctgm(T-t)[\Omega +(d-1)\epsilon^2I^n_d  + bX_s+X_sb^{\top} ])\\
& +\lambda(T-t)^{\top} \kappa(\theta-Y_t) + 2\epsilon^2 \Tr(X \fctgm(T-t) I^n_d\fctgm(T-t))+\frac{1}{2}\Tr(X c^\top \lambda(T-t)  \lambda^\top(T-t) c) \\&+\frac{\epsilon}{2} \Tr(X[c^\top \lambda(T-t) \rho^\top I^n_d \fctgm(T-t)+\fctgm(T-t)I^n_d \rho  \lambda^\top(T-t)c ]) =0.
\end{align*}
By identifying the constant term and the linear terms with respect to $Y_t$ and $X_t$, we get~\eqref{LaplaceGenRiccati} and $ \dot\lambda= -\kappa \lambda+\bar\Lambda, \ \lambda(0)=\Lambda$, which leads to~\eqref{val_lambda} since $\kappa$ is diagonal with positive entries.
By applying Proposition~1.1 of Dieci and Eirola\footnote{We thank Martino Grasselli for pointing us this reference.}~\cite{DieciEirola} to $\Upsilon-\fctgm$, the solution of~\eqref{LaplaceGenRiccati} exists and is well defined for $t\ge 0$. Besides, $\Upsilon-\fctgm$ stays in~$\posm$ by using~\eqref{cond1_ups} and~\eqref{cond2_ups}.

Then, it remains to check that we have indeed~\eqref{laplace}, and it is sufficient to check it for $t=0$.  To do so, we apply It\^{o}'s formula to $M$ and get
$$ dM_s=M_s\left[\Tr(\fctgm(T-s)  [\sqrt{X_s}dW_sI^n_d + I^n_ddW_s^{\top} \sqrt{X_s}] ) + \lambda(T-s)^{\top}c\sqrt{X_s}\left[\bar\rho dZ_s + dW_s\rho \right] \right].$$
Thus, $M$ is a positive local martingale and thus a supermartingale, which gives $M_0\ge \E[M_T]$. To prove that $M_0=\E[M_T]$, we use the argument presented by Rydberg~\cite{Rydberg}. We define $N_t=M_t/M_0$ in order to work with probability measures. We define for ${\bf K}>0$, $\tau_{\bf K}=\inf \{t\ge 0, \Tr(X_t)\ge {\bf K} \}$, $\pi_{\bf K}(x)=\indi{\Tr(x)\le {\bf K}}x+\indi{\Tr(x)\ge {\bf K}} \frac{{\bf K}}{\Tr(x)} x $ for $x \in \posm$ and consider $N^{({\bf K})}_t$ the solution of
  \begin{align*}
 dN^{({\bf K})}_s=&N^{({\bf K})}_s\big\{\Tr(\fctgm(T-s)  [\sqrt{\pi_{\bf K}(X_s)}dW_sI^n_d + I^n_ddW_s^{\top} \sqrt{\pi_{\bf K}(X_s)}] ) \\& + \lambda(T-s)^{\top}c\sqrt{\pi_{\bf K}(X_s)}\left[\bar\rho dZ_s + dW_s\rho \right] \big\}, \\
 N^{({\bf K})}_0=&1.
  \end{align*}
Clearly, $\E[N^{({\bf K})}_T]=1$, and under $\frac{d \Px^{({\bf K})}}{d \Px}=N^{({\bf K})}_T$, $$dW^{({\bf K})}_t=dW_t-2  \sqrt{\pi_{\bf K}(X_t)}\fctgm(T-t)I^n_d-  \sqrt{ \pi_{\bf K}(X_t)}c^\top\lambda(T-t)\rho^\top$$ is a matrix Brownian motion under~$\Px^{({\bf K})}$. %(A VERIFIER)

We now write $\E[N_T]=\E[N_T \indi{\tau_{\bf K}> T}]+\E[N_T \indi{\tau_{\bf K} \le T}]$. By the dominated convergence theorem, we have $\E[N_T \indi{\tau_{\bf K}\le T}] \underset{{\bf K}\rightarrow + \infty}{\rightarrow} 0$. Besides, $\E[N_T \indi{\tau_{\bf K}> T}]=\E[N^{({\bf K})}_T \indi{\tau_{\bf K}> T}]=\Px^{({\bf K})}(\tau_{\bf K}>T)$, and we have to prove that this probability goes to~$1$. To do so, we focus on the following SDE
\begin{align*}d\tilde{X}_t =& (\Omega+(d-1)\epsilon^2I^n_d+ (b+2\epsilon I^n_d\fctgm(T-t)+ \epsilon I^n_d \rho \lambda^\top(T-t)c ) \tilde{X}_t +\tilde{X}_t(b^{\top}+2\epsilon \fctgm(T-t) I^n_d+\epsilon c^\top \lambda(T-t) \rho^\top I^n_d) )dt \\
&+ \epsilon\left(\sqrt{\tilde{X}_t}dW_tI^n_d+I^n_ddW_t^{\top} \sqrt{\tilde{X}_t}\right)
\end{align*}
starting from $\tilde{X}_0=X_0$. We check that $X$ solves before $\tau_{\bf K}$ and under~$\Px^{({\bf K})}$ the same SDE as $\tilde{X}$ under $\Px$. This yields to
$\Px^{({\bf K})}(\tau_{\bf K} > T)=\Px(\inf \{t\ge 0, \Tr(\tilde{X}_t)\ge {\bf K} \} > T)$. Since the SDE satisfied by $\tilde{X}$ is the one of an affine diffusion on~$\posm$, it is well defined for any $t\ge 0$. In particular $\max_{t\in[0,T]} \Tr(\tilde{X}_t)<\infty$ a.s., which gives $\Px(\inf \{t\ge 0, \Tr(\tilde{X}_t)\ge {\bf K} \}>T)\underset{{\bf K}\rightarrow + \infty}{\rightarrow} 1$.
\end{proof}
\begin{remark}\label{rk_cond_lap} The conditions~\eqref{cond1_ups} and~\eqref{cond2_ups} are satisfied for $\Upsilon=0$ if, and only if $-\Gamma \in \posm$ and \begin{equation}\label{cond_ups0}
\forall t\ge 0, \ - \bar \Gamma - \frac{1}{2}c^{\top} \lambda(t) \lambda^{\top}(t) c \in \posm .
  \end{equation}
   Since $|\lambda_i(t)|\le \max(|\Lambda_i|,|\bar \Lambda_i/\kappa_i|)$, we obtain $\lambda(t)^\top\lambda(t)\le  \sum_{i=1}^p\max(\Lambda_i^2,(\bar \Lambda_i/\kappa_i)^2)$. We therefore have $\sum_{i=1}^p\max(\Lambda_i^2,(\bar \Lambda_i/\kappa_i)^2) I_d- \lambda(t)^\top\lambda(t) \in \posm$ and then $\sum_{i=1}^p\max(\Lambda_i^2,(\bar \Lambda_i/\kappa_i)^2) c^{\top}c - c^{\top} \lambda(t) \lambda^{\top}(t) c \in \posm$. Therefore, a sufficient condition for~\eqref{cond_ups0} is
\begin{equation*}
- \bar \Gamma- \frac{1}{2}  \sum_{i=1}^p\max(\Lambda_i^2,(\bar \Lambda_i/\kappa_i)^2) c^\top c \in \posm .
\end{equation*}
\end{remark}
With the Laplace transform~\eqref{laplace}, we have a mathematical tool to check if the process $(X,Y)$ is stationary. This is important for our modeling perspective: unless for some transitory period, one may expect that the factors are stable around some equilibrium. The next proposition give a simple sufficient condition that ensures stationarity. It is proved in Appendix~\ref{ErgoExploSec}.
\begin{proposition}\label{prop_ergo} If $-(b+b^{\top}) \in \posm$ is positive definite, the process $(X,Y)$ is stationary.
\end{proposition}

\begin{remark}
We chose to keep the dynamics of the process $X$ in the space of positive semidefinite matrices as general as possible. Choosing a Wishart specification for $X$ (which corresponds to $\Omega=\epsilon^2\alpha I^n_d, \alpha>0$) does not lead to a significant simplification of the model. While Wishart processes admit an explicit Lapace transform, this is not the case for the process $(X,Y)$ defined by (\ref{SDE_General}). The drift term $\Omega$ allows to account for a mean reversion behavior of the process $X$, we will typically consider a negative mean reversion matrix $b$, in which case we can set $\Omega=-bx^\infty-x^{\infty}b^{\top} $, so that the matrix process $X$ mean reverts to $x^\infty$.
\end{remark}

\subsection{Model definition}\label{ModelDefSec}

\begin{definition}\label{ModelDef}
We assume that $(X_t,Y_t)_{t\geq 0}$ follows~\eqref{SDE_GeneralY} and~\eqref{SDE_General} under a risk-neutral measure. Then, we define the short interest rate by
\begin{equation}\label{ShortRateEq}
    r_t=\varphi+\sum_{i=1}^p Y^i_t+\Tr\left(\gamma X_t\right),
\end{equation}
\noindent with $\varphi \in \R$ and $\gamma \in \symm$.\\
%In the following we denote by $\model$ the stochastic variance-covariance affine term structure model (SCVATSM) defined by~\eqref{ShortRateEq}.
\end{definition}

\noindent  From Proposition~\ref{FourLapGeneralProp}, we easily get the following result on the Zero-Coupon bonds.
\begin{nt_corollaire}\label{Cor_ZCB}\textbf{Bond reconstruction formula}. Let $0\le t \le T$ and $P_{t,T}=\E[\exp(-\int_t^T r_sds)|\mathcal{F}_t]$ denote the price at time~$t$ of a zero-coupon bond with maturity~$T$. Let us assume that
\begin{equation}\label{cond_gamma} \gamma -\frac{1}{2} \left( \sum_{i=1}^p \frac{1}{\kappa_i^2} \right) c^\top c \in \posm.
\end{equation}
Then, by using Remark~\ref{rk_cond_lap}, $P_{t,T}$ is given by
\begin{equation}\label{zc}
    P_{t,T}=\exp(A(T-t)+\Tr(D(T-t)X_t)+B(T-t)^{\top} Y_t),
\end{equation}
with $A(t)=\eta(t)-\varphi t$, $D(t)=g(t)$ and $B(t)=\lambda(t)$, where $(\eta,g,\lambda)$ is the solution of~\eqref{LaplaceGenRiccati} with~\eqref{val_lambda}, $\Lambda=0$, $\Gamma=0$, $\bar \Gamma= -\gamma$ and $\bar \Lambda =-\mathbf{1}_p$ (i.e. $\bar \Lambda_i=-1$ for $1\le i\le p$). In particular, we have $-D(T-t) \in \posm$.
\end{nt_corollaire}

\noindent Let us make general comments on the model.  It is close but slightly different from the model proposed in the PhD Thesis of Bensusan~\cite{HB}. Nonetheless, our presentation as a perturbation of the LGM model enables us to have a better understanding of the model parameters. Thus, the vector process $Y$ can be interpreted as in the LGM model, meaning it is assumed to be the main driver of the yield curve. The individual factors are viewed as principal components movements of the yield curve. We chose to specify the model such that the matrix process $X$ admits a similar interpretation. Typically we will consider  $\Omega=-bx^\infty-x^{\infty}b^{\top} $ with $b$ symmetric negative to have a mean-reversion toward a given covariance matrix $x^\infty$. The parameter $\epsilon$ measures the level of the perturbation around the LGM. The matrix process $X$ plays the role of a stochastic variance-covariance matrix of the main movements of the yield curve. It is possible to define the diffusion parameter $c$ such that the diagonal factors of the matrix $X$ play the role of the instantaneous stochastic variance of the yield curve movement and the off-diagonal terms play the role of the instantaneous covariance between two yield curve movements. The vector $\rho$ is a correlation parameter between the processes $Y$ and $X$. In a first approximation\footnote{Note that this is not completely true, even in the simple LGM model. One important characteristic of short rate/factorial interest rates model is that the yield curve depends not only on the (stochastic) state variables of the model, but also on the volatility of the state variables. Therefore the volatility factors $X$ appear in the payoff of interest rates options.}, interest rates options are options on linear combinations of the factors $Y$, and instantaneous variance of these linear combinations are linear combinations of the factors $X$. Therefore, the correlation parameter $\rho$ will  drive the skew of interest rates options. We now make more precise comments on the model.

\begin{itemize}
\item In order to keep the same factors as in the LGM, one would like to take $\gamma=0$. However, this choice is possible only if the perturbation around the LGM is small enough provided that $-(b+b^\top)$ is positive definite, see Remark~\ref{rk_cond_lap_2}. Besides, even if $P_{t,T}$ may be well defined for $T-t$ small enough, it would be then given by the same formula, and therefore the yield curve dynamics depends anyway on the factor~$X$.
\item In order to have a clear interpretation of the volatility factor~$X$ on the factor $Y$, a possible choice is to consider $d=q \times p$ with $q \in \mathbb{N}^*$ and $c_{i,j}=\indi{(i-1) \times p < j \le i\times p}$. Thus, from~\eqref{crochets_1}, the principal matrix $(X_{k,l})_{ (i-1) \times p <k,l \le i\times p}$ rules the instantaneous quadratic variation of the factor $Y_i$ while the submatrix  $(X_{k,l})_{ (i-1) \times p <k \le i\times p,(j-1) \times p <l \le j\times p }$ rules the instantaneous covariation between the factors $Y^i$ and $Y^j$.
\item The model does not prevent from having a negative short rate or from having $\E[|P_{t,T}|^k]=\infty $ for any $k>0$, unless we consider the degenerated case ($p=0$) where the yield curve is driven by the volatility factors $X$ and the factors $Y$ are null. This particular model has been studied by Gnoatto in~\cite{Gnoatto}.
\item  Affine Term Structure models generally consider constant parameters that are fixed over a large period and reflect the market behaviour, while the current value of factors are fitted to market data. This is why we consider constant parameters here. However, in order to fit exactly Zero-Coupon Bond prices, it is possible to take a time-dependent function $\varphi$ while keeping the tractability of the model.
\end{itemize}

\begin{remark}\label{rk_cond_lap_2}
The condition~\eqref{cond_gamma} is sufficient to get that~$P_{t,T}$ is well-defined. However, this condition does not depend on $\epsilon$ while we know that for $\epsilon=0$, $P_{t,T}$ is well-defined since $X$ is deterministic and $Y$ is a Gaussian process. We can get a complementary sufficient condition when $-(b+b^\top)$ is positive definite, which is a reasonable assumption since it leads to a stationary process by Proposition~\ref{prop_ergo}. In this case, there exists $\mu>0$ such that $-(b+b^\top)-\mu I_d \in \posm$. By using Proposition~\ref{FourLapGeneralProp} with $\Upsilon=\frac{\mu}{4\epsilon^2}I_d$,  we get that~\eqref{cond2_ups} is satisfied if we have
$$\forall t\ge 0,  \frac{\mu^2}{8\epsilon^2}I_d- \frac{\mu}{8\epsilon}(I^n_d\rho \lambda^{\top} c+I^n_d\rho \lambda^{\top} c)-\frac{1}{2}c^{\top} \lambda\lambda^{\top} c+\gamma \in \posm. $$
Since for $t\ge0$, $\lambda(t)$ takes values in a compact subset of~$\R^p$, there is $\epsilon_0>0$ such that this condition is satisfied for any $\epsilon \in (0,\epsilon_0)$. 
\end{remark}

\begin{remark} Let $a \in \genm$, and consider the model $r_t=\varphi+\sum_{i=1}^p Y^i_t+\Tr\left( \tilde{\gamma} \tilde{ X_t}\right)$ with
 \begin{align*}
Y_t &=y +\int_0^t\kappa(\theta -Y_s) ds + \int_0^t \tilde{c}\sqrt{\tilde{X}_s}\left[\sqrt{1-|\tilde{\rho}|^2} dZ_s + dW_s \tilde{\rho} \right]\\
\tilde{X}_t &=\tilde{x} +\int_0^t\left( \tilde{\Omega} +(d-1)\epsilon^2a^\top a  + \tilde{b} \tilde{X}_s+\tilde{X}_s \tilde{b}^{\top} ) \right )ds + \epsilon\int_0^t\sqrt{\tilde{X}_s}dW_sa + a^\top dW_s^{\top} \sqrt{\tilde{X}_s},
 \end{align*}
and $\tilde{\gamma} \in \symm$, $\tilde{x},\tilde{\Omega} \in \posm$, $\tilde{c},\tilde{b} \in \genm$, $\tilde{\rho}\in \R^d$ such that $|\tilde{\rho}|\le 1$. This model may seem a priori more general, but this is not the case. In fact, let $n$ be the rank of~$a$ and $u\in \genm$ be an invertible matrix such that $a^\top a= (u^{-1})^\top I^n_d (u^{-1})$.  Then, $X_t=u^\top \tilde{X}_t u$ solves
$$dX_t=[ \Omega  + (d-1)\epsilon^2I^n_d + bX_t+X_tb^\top ] dt + u^\top \sqrt{\tilde{X}_t}dW_t a u + u^\top a^\top dW_t^{\top} \sqrt{\tilde{X}_t} u,$$
with $b=u^\top \tilde{b} (u^{-1})^\top$, $\Omega=u^\top \tilde{\Omega} u \in \posm$ and starting from $x=u^\top \tilde{x} u \in \posm$. After some calculations, we obtain
$\langle d(Y_t)_m,d(Y_t)_{m'} \rangle=(\tilde{c}\tilde{X}_t\tilde{c}^\top)_{m,m'}dt=(cX_tc^\top)_{m,m'}dt$ with $c=\tilde{c}(u^{-1})^\top$; $\langle d(X_t)_{i,j},d(X_t)_{k,l} \rangle=\epsilon^2 \left[(X_t)_{i,k}\indi{j=l\le n}+(X_t)_{i,l}\indi{j=k\le n} + (X_t)_{j,k}\indi{i=l\le n}+  (X_t)_{j,l}\indi{i=k\le n} \right]dt$ and
\begin{align*}
 \langle d(Y_t)_m,d(X_t)_{i,j} \rangle &=\epsilon \left[ (u^\top \tilde{X}_t\tilde{c}^\top )_{m,i} (u^\top a^\top \tilde{\rho})_j + (u^\top \tilde{X}_t\tilde{c}^\top )_{m,j} (u^\top a^\top \tilde{\rho})_i \right] dt \\
&=\epsilon \left[ (X_t {c}^\top )_{m,i} (u^\top a^\top \tilde{\rho})_j + (X_t {c}^\top )_{m,j} (u^\top a^\top \tilde{\rho})_i  \right]dt.
\end{align*}
Since the law of $(X,Y)$ is characterized by its infinitesimal generator, we can assume without loss of generality that $\tilde{\rho} \in \ker(u^\top a^\top)^\bot=\textup{Im}(au)$. Therefore, there is $\rho' \in \R^d$ such that $\tilde{\rho}=au \rho'$, and we set $\rho_i=\rho'_i$ for $i\le n$ and $\rho_i=0$ for $n< i\le d$. We have $|\rho|^2=(\rho')^\top I^n_d \rho' =|\tilde{\rho}|^2\le 1$, and therefore $(X,Y)$ follows the same law as the solution of~\eqref{SDE_GeneralY} and~\eqref{SDE_General}, and we have $r_t=\varphi+\sum_{i=1}^p Y^i_t+\Tr(\gamma X_t)$ with $\gamma=u^{-1}\tilde{\gamma} (u^{-1})^\top$.
\end{remark}

\subsection{Change of measure and Laplace transform}\label{MeasLapSec}

In the fixed income market, the pricing of vanilla products is often (if not always) made under a suitably chosen equivalent martingale measure different from the risk-neutral measure. It is thus important to characterize the distribution of the underlying state variables under these measures. The forward-neutral measures are probably the most important example of such pricing measures. In this paragraph, we will see that the dynamics of the factors remains affine and keeps the same structure under the forward measures.

\subsubsection{Dynamics under the forward-neutral measures}\label{MeasSec}
We assume that the condition~\eqref{cond_gamma} holds.
Let $Q^U$ denote the $U$-forward neutral probability, which is defined on~$\mathcal{F}_U$ by
$$\frac{d Q^U}{d \mathbb{P}}=\frac{e^{-\int_0^Ur_sds}}{P_{0,U}}.$$
This is  the measure associated with the numeraire $P_{t,U}$. It comes from the martingale property of discounted asset prices that for $t\in (0,U)$,
\begin{eqnarray*}
\frac{d\left(e^{-\int_0^tr_sd_s}P_{t,U}\right) }{e^{-\int_0^tr_sd_s}P_{t,U}}&=&2\epsilon\Tr(D(U-t)\sqrt{X_t}dW_tI^n_d)+B(U-t)^{\top} c\sqrt{X_t}dW_t\rho+\bar\rho B(U-t)^{\top} c\sqrt{X_t}dZ_t \\
 &=&\Tr([2\epsilon I^n_d D(U-t)\sqrt{X_t} + \rho B(U-t)^{\top} c\sqrt{X_t}]dW_t)+\bar\rho B(U-t)^{\top} c\sqrt{X_t}dZ_t.
\end{eqnarray*}

\noindent From Girsanov's theorem, the processes
\begin{eqnarray*}
  dW^U_t &=& dW_t-\sqrt{X_t}(2\epsilon D(U-t)I^n_d+c^{\top} B(U-t)\rho^{\top} )dt \\
  dZ^U_t &=& dZ_t-\bar\rho\sqrt{X_t}c^{\top} B(U-t)dt
\end{eqnarray*}
are respectively matrix and vector valued Brownian motions under $Q^U$ and are independent. This yields to the following dynamics for $Y$ and $X$ under $Q^U$:
\begin{eqnarray}
  \label{sigF}dX_t &=& (\Omega+(d-1)\epsilon^2I^n_d+ b^U(t)X_t +X_t (b^U(t))^{\top} )dt+ \epsilon\left(\sqrt{X_t}dW^U_tI^n_d+I^n_d(dW^U_t)^{\top} \sqrt{X_t}\right) \label{X_Ufwd}\\
  \label{YF} dY_t &=& \kappa(\theta-Y_t)dt+ cX_tc^{\top} B(U-t)dt+2\epsilon cX_tD(U-t)I^n_d\rho dt+c\sqrt{X_t}(dW^U_t\rho+\bar\rho dZ^U_t),\label{Y_Ufwd}
\end{eqnarray}
with  $b^U(t)=b+2\epsilon^2I^n_dD(U-t)+\epsilon I^n_d\rho B(U-t)^{\top} c$. %Note that under $Q^U$ the process $X$ is still a Wishart process with a time-dependent mean-reverting matrix.

\subsubsection{Laplace transforms}\label{LaplaceSec}

We are now interested in calculating the law of $(X_T,Y_T)$ under the $U$-forward measure for $T\le U$. More precisely, we calculate $\E^{Q^U}\Big[\exp(\Tr(\Gamma X_T)+\Lambda^{\top} Y_T)|\mathcal{F}_t\Big]$ for $t\in[0,T]$ by using again Proposition~\ref{FourLapGeneralProp}. We assume that condition~\eqref{cond_gamma} holds and have
\begin{align*}
&\E^{Q^U}\Big[\exp(\Tr(\Gamma X_T)+\Lambda^{\top} Y_T)|\mathcal{F}_t\Big] \\
&=\frac{1}{P_{t,U}} \E\Big[\exp(\Tr(\Gamma X_T)+\Lambda^{\top} Y_T -  (U-t)\varphi - \int_t^U \mathbf{1}_p^\top Y_s ds -\int_t^U \Tr(\gamma X_s) ds )|\mathcal{F}_t\Big] \\
&=\frac{\E \Big[\exp(\Tr(\Gamma+D(U-T) X_T)+ (\Lambda+B(U-T))^{\top} Y_T + A(U-T) - (T-t)\varphi  - \int_t^T \mathbf{1}_p^\top Y_s ds -\int_t^T \Tr(\gamma X_s) ds ) |\mathcal{F}_t\Big]}{\exp(A(U-t)+\Tr(D(U-t)X_t)+B(U-t)^{\top} Y_t)}.
\end{align*}
We consider $\Gamma \in \symm$ and $\Lambda \in \R^p$ such that
$$ - \Gamma  \in \posm \text{ and } |\Lambda_i|\le e^{-\kappa_i(U-T)}/\kappa_i, \text{ for } 1\le i \le p,$$
in order to have $|\Lambda_i+B_i(U-T)|\le 1/\kappa_i$ and  $- (\Gamma+D(U-T)) \in \posm$. By Proposition~\ref{FourLapGeneralProp}, condition~\eqref{cond_gamma} and Remark~\ref{rk_cond_lap}, we get that the expectation is finite and that
\begin{equation}\label{Lap_Ufwd} \E^{Q^U}\Big[\exp(\Tr(\Gamma X_T)+\Lambda^{\top} Y_T)|\mathcal{F}_t\Big]=\exp(A^U(t,T)+\Tr(D^U(t,T)X_t)+B^U(t,T)^{\top} Y_t),
\end{equation}
with $F^U(t,T)=\tilde{F}(T-t)+F(U-T)-F(U-t)$ for $F\in\{A,D,B\}$, where $(\tilde{B},\tilde{D},\tilde{A})$ is the solution of~\eqref{LaplaceGenRiccati} with $\tilde{B}(0)=\Lambda+B(U-T)$, $\tilde{D(0)}=\Gamma+D(U-T)$, $\tilde{A}(0)=0$, $\bar \Lambda=\mathbf{1}_p$ and $\bar\Gamma=-\gamma$.

\begin{nt_corollaire}Let~\eqref{cond_gamma} hold.
For $\Gamma \in \symm$ and $\Lambda \in \R^p$ such that $ - \Gamma  \in \posm$  and $|\Lambda_i|\le e^{-\kappa_i(U-T)}/\kappa_i$  for $1\le i \le p$, $\E^{Q^U}\Big[\exp(\Tr(\Gamma X_T)+\Lambda^{\top} Y_T)|\mathcal{F}_t\Big]<\infty$ a.s. for any $t\in[0,T]$ and is given by~\eqref{Lap_Ufwd}.
\end{nt_corollaire}

Let us mention that in practice, the formula above for  $A^U(t,T)$, $D^U(t,T)$ and $B^U(t,T)$ requires to solve two different ODEs. It may be more convenient to use the following one that can be easily deduced from dynamics of $(X,Y)$ under the $U$-forward measure:
\begin{equation}
\label{LaplaceRiccati}
\left\{
\begin{array}{rl}
\frac{\partial B^U}{\partial t}(t,T)&= \kappa^{\top} B^U(t,T), \ B^U(T,T)=\Lambda, \\
\\
-\frac{\partial D^U}{\partial t}(t,T)&=2\epsilon^2D^U I^n_dD^U+D^U (b^U(t)+\epsilon I^n_d\rho (B^U)^{\top} c)+(b^U(t)+\epsilon I^n_d\rho (B^U)^{\top} c)^{\top} D^U+\frac{1}{2}c^{\top} B^U(B^U)^{\top} c\\
\\
 & +c^{\top} B^UB(U-t)^{\top} c+\epsilon D(U-t)I^n_d\rho (B^U)^{\top} c +\epsilon c^{\top} B^U\rho^{\top} I^n_dD(U-t), \ D^U(T,T)=\Gamma,\\
-\frac{\partial A^U}{\partial t}(t,T)&=B^U(t,T)^{\top} \kappa\theta+\mathrm{Tr}\left(D^U(t,T)(\Omega+\epsilon^2(d-1)I^n_d)\right), \ A^U(T,T)=0.\\
%% \\
%% D^U(T,T)=\Gamma, \quad B^U(T,T)=\Lambda,\quad A^U(T,T)=0.
\end{array}
\right.
\end{equation}

%***********************************************************************************************************************************************************
%***********************************************************************************************************************************************************

\section{Expansion of the volatility smile around the LGM}\label{Sec_PriceExp}

The goal of this section is to provide the asymptotic behaviour of the Caplet and Swaption prices when  the volatility parameter $\epsilon$ is close to zero. The practical interest of these formulas is to give quickly a proxy for these prices. Thus, they give a tool to calibrate the model parameters to the smile.  Let us mention here that expansions of Gram-Charlier type can be also be applied to price caplets and swaptions thanks to the affine structure of the model, see for example Collin-Dufresne and Goldstein~\cite{DGE} and Tanaka et al.~\cite{KT}. Some numerical examples are presented in~\cite{Palidda} for the pricing of caplets in our model. Here, we only present the expansion with respect to $\epsilon$ since it is in accordance with our presentation of the model as a perturbation of the LGM.

 The arguments that we use in this section to obtain the expansion have been developed in the book of Fouque et al. \cite{Fouque}. They rely on an expansion of the infinitesimal generator with respect to~$\epsilon$.  Recently, this technique was applied by Bergomi and Guyon~\cite{Bergomi} to provide approximation under a multi factor model for the forward variance. Here, we have to take into account some specific  features of the fixed income and work under the appropriate probability measure to apply these arguments.
Not surprisingly the zero order term in the expansion is exactly the volatility of the LGM with a time-dependent variance-covariance matrix. More interestingly the higher order terms allow to confirm the intuitions on the role of the parameters and factors that determine the shape and dynamics of the volatility.

Last, we have to mention that the calculations presented in this section are rather formal. In particular, we implicitly assume that the caplet and swaption prices are smooth enough and admit expansions with respect to~$\epsilon$.  A rigorous proof of these expansions is beyond the scope of this paper.

\subsection{Price and volatility expansion for Caplets}\label{ExpansionCaplets}
 From~\eqref{CapletPricing}, the only quantity of interest in order to understand the Caplets volatility cube is what we call the forward Caplet price

\begin{equation*}%\label{CallForwardLibor}
    \mathrm{FCaplet}(t,T,\delta)=\E^{T+\delta}\left[\left(L_T(T,\delta)- K\right)^+|\mathcal{F}_t\right],
\end{equation*}

\noindent which can be rewritten as a call option on the forward zero coupon bond $\frac{P_{t,T}}{P_{t,T+\delta}}$

\begin{equation*}%\label{CallForward}
    \mathrm{FCaplet}(t,T,\delta)=\frac{1}{\delta}\E^{T+\delta}\left[\left(\frac{P_{T,T}}{P_{T,T+\delta}}-(1+\delta K)\right)^+\bigg|\mathcal{F}_t\right].
\end{equation*}

\noindent Since $(X,Y)$ is a Markov process,  $\mathrm{FCaplet}(t,T,\delta)$ is a function of $(X_t,Y_t)$ and therefore we can define the forward price function

\begin{equation}\label{CallForward}
    P(t,x,y)=\E^{T+\delta}\left[\left(\frac{P_{T,T}}{P_{T,T+\delta}}-(1+\delta K)\right)^+\bigg|X_t=x, Y_t=y\right].
\end{equation}
The goal of Subsection~\ref{ExpansionCaplets} is to obtain the second order expansion~\eqref{expansionPcaplet} of~$P$ with respect to $\epsilon$.

\subsubsection{A convenient change of variable}

We want to get an expansion of the caplet price with respect to~$\epsilon$. To do so, we need a priori to get an expansion  to~$\epsilon$  of the infinitesimal generator of the process $(X,Y)$ under the probability $Q^{T+\delta}$. However, we can make before a change of variable that simplifies this approach. Thus, we define
$$\varmaj_t=\Delta A(t,T,\delta)+\Tr(\Delta D(t,T,\delta) X_t)+\Delta B(t,T,\delta)^{\top} Y_t,$$
with
\begin{eqnarray*}
    \Delta A(t,T,\delta)&=&A(t,T)-A(t,T+\delta)\\
    (\Delta B,\Delta D)(t,T,\delta)&=&(B,D)(T-t)-(B,D)(T+\delta-t)
\end{eqnarray*}
Thus, we have $\frac{P_{t,T}}{P_{t,T+\delta}}=e^{H_t}$. It is well known that $\frac{P_{t,T}}{P_{t,T+\delta}}$ is a martingale under~$Q^{T+\delta}$, see e.g. Proposition 2.5.1 in Brigo and Mercurio~\cite{BM}. Thus, we get by It\^{o} calculus from~\eqref{X_Ufwd} and~\eqref{Y_Ufwd} that $(X,H)$ solve the following SDE
\begin{eqnarray}
dX_t &=& (\Omega+\epsilon^2(d-1) I^n_d + b^{T+\delta}(t)X_t +X_t (b^{T+\delta}(t))^{\top})dt+ \epsilon\sqrt{X_t}dW^{T+\delta}_tI^n_d +\epsilon I^n_d (dW^{T+\delta}_t)^{\top}\sqrt{X_t} ,\nonumber \\
  \nonumber d \varmaj_t &=& -\frac{1}{2}\left(\Delta B^{\top}cX_tc^{\top}\Delta B+4\epsilon^2\Tr(\Delta DI^n_d\Delta DX_t)+2\epsilon(\Delta B^{\top}c X_t\Delta DI^n_d\rho)\right)dt \\
  && +\Delta B^{\top}c\sqrt X_t (dW^{T+\delta}_t\rho +\bar\rho dZ^{T+\delta}_t)+2\epsilon\Tr(\Delta D\sqrt X_t dW^{T+\delta}_t I^n_d). \label{sigXL}
\end{eqnarray}
Therefore,  $P(t,x,y)=\E^{T+\delta}\left[\left(e^{H_T}-(1+\delta K)\right)^+|X_t=x, Y_t=y\right]$ only depends on~$(x,y)$ through $(x,h)$ where $\varmin = \Delta B(t,T,\delta)^{\top}y+\Tr\left(\Delta D(t,T,\delta) x\right)+\Delta A(t,T,\delta)$, and we still denote by a slight abuse of notations
$$  P(t,x,h)=\E^{T+\delta}\left[\left(e^{H_T}-(1+\delta K)\right)^+|X_t=x, H_t=h\right].$$

\noindent Let us emphasize that this change of variable is crucial in order to apply an expansion procedure similar to the one of Bergomi and Guyon~\cite{Bergomi}. It allows to reduce the dimensionality of the underlying state variable. The variable $\varmaj$ is one-dimensional and it is the only variable that appears in the payoff of the caplet. Though this is obvious from the definition of the model, we insist on the fact that the implied volatility of caplets is a function of the factors $X$ only. This appears clearly in the SDE (\ref{sigXL}), $\varmaj_t$ can be viewed as continuous version of the forward Libor rate and its volatility depends on the factors $X$ only.

\subsubsection{Expansion of the price}

From the SDE~\eqref{sigXL},~\eqref{crochets_1},~\eqref{crochets_2} and~\eqref{crochets_3}, we get the following PDE representation of $P$:

\begin{eqnarray*}
  \partial_t P + \cL(t) P   &=&0 \\
  P(T,x,\varmin)&=& (e^\varmin-(1+\delta K))^+
\end{eqnarray*}

\noindent where $\cL(t)$ is the infinitesimal generator of~\eqref{sigXL}. We assume that $P$ admits a second order expansion \begin{equation}\label{expansionPcaplet}P=P_0+\epsilon P_1+\epsilon^2 P_2+o(\epsilon^2).
\end{equation}
Our goal is to calculate in a quite explicit way the value of $P_0$, $P_1$ and $P_2$. We assume in our derivations that these functions $P_0$, $P_1$ and $P_2$ are smooth enough. To determine the value of $P_0, P_1$ and $P_2$, we proceed as Bergomi and Guyon~\cite{Bergomi} and make an expansion of the generator $\cL(t)=\cL_0(t)+\epsilon \cL_1(t) +\epsilon^2 \cL_2(t) + \dots$ in order to obtain the PDEs satisfied by $P_0, P_1$ and $P_2$. Namely, we obtain
\begin{eqnarray*}
    \partial_t P_0 +\cL_0(t)P_0 &=& 0,\quad P_0(T,x,\varmin)=(e^\varmin-(1+\delta K))^+,\\
    \partial_t P_1 +\cL_0(t)P_1 + \cL_1(t)P_0&=&0,\quad P_1(T,x,\varmin)=0,\\
    \partial_t P_2 +\cL_0(t)P_2 + \cL_2(t)P_0+\cL_1(t) P_1&=&0,\quad P_2(T,x,\varmin)=0.
\end{eqnarray*}
Thus, we can solve first the PDE for $P_0$, then for $P_1$ and so on. Let $\mathrm{BS}(h,v)=\E\left[\left(\exp\left(h-\frac{1}{2}v+\sqrt{v}G\right)-(1+\delta K)\right)^+\right]$ with $G\sim N(0,1)$ denote the Black-Scholes price with realized volatility~$v$. We obtain easily that
\begin{equation*}
    P_0(t,x,\varmin)=\mathrm{BS}(\varmin,v(t,T,\delta,x)),
\end{equation*}

\noindent with 

\begin{align}
    v(t,T,\delta,x)&=\int_t^{T}\Delta B(u,T,\delta)^{\top}c X^0_{u-t}(x) c^{\top}\Delta B(u,T,\delta) du, \label{vcaplet}\\
\label{defX0sx}X^0_s(x)&=e^{bs}\left(x+\int_0^se^{-bu}\Omega e^{-b^{\top} u}du \right)e^{b^{\top}s}.
\end{align}
\noindent The higher order terms are given by\footnote{The details of these simple but tedious   calculations are available online {\tt http://arxiv.org/abs/1412.7412} in the first draft of this paper for the caplets and swaptions.}
\begin{eqnarray}
    \label{FormulaCapletExpOrder1}P_1(t,x,\varmin)&=&\left(c_1(t,T,\delta,x)(\partial^3_\varmin-\partial^2_\varmin) +c_2(t,T,\delta,x)(\partial^2_\varmin-\partial_\varmin)\right) P_0(t,x,\varmin)
\end{eqnarray}
and
\begin{align}
    \nonumber P_2(t,x,h)=&\Bigg[\left(d_1(t,T,\delta,x)(\partial^2_\varmin-\partial_\varmin)^2+d_2(t,T,\delta,x)(\partial^2_\varmin-\partial_\varmin)\partial_\varmin +d_3(t,T,\delta,x)(\partial^2_\varmin-\partial_\varmin)\right) \label{FormulaCapletExpOrder2} \\&+\bigg( e_1(t,T,\delta,x) (\partial_h^2-\partial_h)^2 \partial_h^2+ e_2(t,T,\delta,x) (\partial_h^2-\partial_h)^2 \partial_h  + e_3(t,T,\delta,x) (\partial_h^2-\partial_h)^2 \\&+e_4(t,T,\delta,x) (\partial_h^2-\partial_h)\partial_h^2+e_5(t,T,\delta,x) (\partial_h^2-\partial_h)\partial_h+e_6(t,T,\delta,x) (\partial_h^2-\partial_h) \bigg) \Bigg]P_0(t,x,h). \nonumber
\end{align}
The coefficients $c_i$, $d_i$ and $e_i$ are given in Appendix~\ref{CapletsExpansionCalcul}. We recall that the derivatives $\partial^i_hP_0$ of $P_0$ with respect to $h$ can be calculated explicitly, so that the expansion is very efficient from the point of view of the computational time, see Section~\ref{sec_comp_meth}.

\begin{remark} It is easy to obtain then the expansion $v_{\rm Imp} =v_0 +\epsilon v_1 +\epsilon^2 v_2  +o(\epsilon^2)$ of the implied volatility  defined by $\delta \mathrm{FCaplet}(t,T,\delta)=\mathrm{BS}(\varmin,v_{\rm Imp}  )$. We obtain as expected $v_0 = v(t,T,\delta,x)$ and 
\begin{equation}
\label{CapletVolExpOrder1}\frac{v_1}{2} =c_2(t,T,\delta,x)+c_1(t,T,\delta,x)\left(\frac12-\frac{h-\log(1+\delta K)}{  v_0}\right).
\end{equation}
Since neither $c_1$ nor $c_2$ depend on the strike,  the skew is  at the first order in $\epsilon$  proportional  to $c_1$, that is at its turn a linear function of~$\rho$. We have in particular a flat smile at the first order when $\rho=0$, as one may expect. 
\end{remark}

\subsection{Price and volatility expansion for Swaptions}\label{ExpansionSwaptions}

\noindent From~\eqref{SwaptionPricing}, the only quantity of interest in order to understand the swaptions volatility cube is what we call the annuity-forward swaption price

\begin{equation*}%\label{CallForward}
    \mathrm{AFSwaption}(t,T,m,\delta)=\E^{A}\left[\left(S_t(T,m,\delta)-K\right)^+|\mathcal{F}_t\right].
\end{equation*}

\noindent It is standard to view swaptions as a basket option of forward Libor rates with stochastic weights, we have

\begin{eqnarray} 
  \label{SwapAsBasket}S_t(T,m,\delta) &=& \sum_{i=1}^m\omega^i_tL_t(T+(i-1)\delta,T+i\delta)\\
  \label{Weights}\omega^i_t &=& \frac{P_{t,T+i\delta}}{\sum_{i=1}^mP_{t,T+i\delta}}.
\end{eqnarray}

\noindent The difficulty here comes from the fact that forward Libor rates, and the stochastic weights are complicated functions of the state variables $(X,Y)$. The first implication is that the change of measure between $\Px$ and $Q^{A}$ is also complicated and the dynamics of the state variables under this new measure is quite unpleasant to work with. The second implication is that we cannot directly operate a convenient change of variable as we did for caplets. In order to derive an expansion for swaptions we thus proceed stepwise. First, we use a standard approximation that freezes the weights at their initial value (see for example Brigo and Mercurio~\cite{BM} p.~239, d'Aspremont \cite{DBLP:journals/corr/cs-CE-0302034} and Piterbarg~\cite{VP}). This is justified by the fact the variation of the weights is less important than the variation of the forward Libor rates\footnote{To the best of our knowledge there have been very few attempts to quantify either theoretically or numerically this statement. In \cite{DBLP:journals/corr/cs-CE-0302034} d'Aspremont investigates the accuracy of the approximation for pricing swaptions in the log-normal BGM model, he shows that the approximation is less efficient for long maturities and long tenors.}. Second, we  use a similar approximation for the swap rate. Thus, the approximated swap rate is an affine function of the underlying state variables, which enables us to take advantage of the affine structure of the model. Let us mention that this technique is similar to the quadratic approximation of the swap rate proposed by Piterbarg in~\cite{VP}. Finally we perform our expansion on the affine approximation of the swap rates and obtain the second order expansion~\eqref{expansionPswaption}, which is the main result of Subsection~\ref{ExpansionSwaptions}.

\subsubsection{Dynamics of the factors under the annuity measure}

The annuity measure knowing the information up to date $t$, $Q^A|\mathcal F_t$ is defined   by

\begin{equation*}
    \frac{dQ^A}{d\Px} \bigg|_{\mathcal F_t}=e^{-\int_t^Tr_s ds}\frac{A_T(T,m,\delta)}{A_t(T,m,\delta)}.
\end{equation*}

\noindent It comes from the martingale property of discounted asset prices under the risk neutral measure that

\begin{equation}\label{MeasureChangeAnnuity}
    \frac{d\left(e^{-\int_0^t r_s ds}A_t(T,m,\delta)\right)}{e^{-\int_0^t r_sds}A_t(T,m,\delta)}=\sum_{i=1}^m\omega^i_t\left(B(T+i\delta-t)^{\top}c\sqrt{X_t} (dW_t\rho+\bar\rho dZ_t)+2\epsilon\Tr\left(D(T+i\delta-t)\sqrt{X_t} dW_t I^n_d\right)\right).
\end{equation}

\noindent From Girsanov's theorem, the change of measure is given by

\begin{eqnarray*}
  dW^{A}_t &=& dW_t-\sqrt{X_t}\left(2\epsilon\sum_{i=1}^m\omega^i_tD(T+i\delta-t)I^n_d+c^{\top}B(T+i\delta-t)\rho^{\top}\right)dt,\\
  dZ^{A}_t &=& dZ_t-\bar\rho\sqrt{X_t}c^{\top}\sum_{i=1}^m\omega^i_tB(T+i\delta-t)dt.
\end{eqnarray*}

\noindent This allows us to calculate from~\eqref{SDE_GeneralY} and~\eqref{SDE_General} the dynamics of the state variables under the annuity measure $Q^A$:
\begin{eqnarray}
 \label{YAnnuity} dY_t &=& \left(\kappa(\theta-Y_t)+cX_tc^{\top}\sum_{k=1}^m\omega^k_tB(T+k\delta-t)+2\epsilon cX_t\sum_{k=1}^m\omega^k_tD(T+k\delta-t)I^n_d\rho \right)dt \\
&& +c \sqrt{X_t}(\bar \rho dZ^A_t + dW^A_t \rho), \nonumber \\
  \label{sigAnnuity}dX_t &=& (\Omega+\epsilon^2(d-1) I^n_d + b^{A}(t)X_t +X_t (b^{A}(t))^{\top})dt+ \epsilon\left(\sqrt{X_t} dW^A_tI^n_d+I^n_d(dW^A_t)^{\top}\sqrt{X_t}\right),
\end{eqnarray}

\noindent where $b^{A}(t)=b+\epsilon I^n_d\rho \sum_{k=1}^m\omega^k_tB(T+k\delta-t)^{\top}c+2\epsilon^2I^n_d\sum_{k=1}^m\omega^k_tD(T+k\delta-t).$

\subsubsection{An affine approximation of the forward swap rate}

The forward swap rate is a martingale under the annuity measure $Q^{A}$. Therefore, we can only focus on the martingale terms when applying It\^o's formula to $\frac{P_{t,T}-P_{t,T+m \delta}}{\sum_{i=1}^mP_{t,T+i\delta}}$, and we get from~\eqref{zc} that
\begin{align}
 \label{SAnnuity}
&   dS_t(T,m,\delta) =  \\
&= \left[ \omega^0_t B(T-t)^{\top} -  \omega^m_t B(T+m\delta-t)^{\top} - S_t(T,m,\delta) \sum_{k=1}^m\omega^k_t B(T+k\delta-t)^{\top} \right] c\sqrt{X_t} (dW^A_t\rho+\bar\rho dZ^A_t)\nonumber  \\
& \ + 2\epsilon  \Tr\left( \left[ \omega^0_t D(T-t) -  \omega^m_t D(T+m\delta-t) - S_t(T,m,\delta) \sum_{k=1}^m\omega^k_t D(T+k\delta-t) \right]\sqrt{X_t} dW^A_t  I^n_d\right) \nonumber
\end{align}

\noindent By a slight abuse of notations,  we will now drop the $(T,m,\delta)$ dependence of the swap rate and simply denote by $S_t$ its time $t$ value. We now use the standard approximation that consists in  freezing the weights $\omega^k_t$ and the value of the swap rate $S_t$ in the right-hand side to their value at zero. We then have
\begin{equation}\label{SAnnuityApprox}
   dS_t=B^S(t)^{\top}c\sqrt{X_t} (dW^A_t \rho+\bar\rho dZ^A_t)+2\epsilon\Tr\left(D^S(t)\sqrt{X_t} dW^A_t I^n_d\right),
\end{equation}

\noindent where

\begin{eqnarray*}
  (B,D)^S(t) &=& \omega^0_0(B,D)(T-t)-\omega^m_0(B,D)(T+m\delta-t)-S_0(T,m,\delta)\sum_{k=1}^m\omega^k_0(B,D)(T+k\delta-t).
\end{eqnarray*}
These coefficients are time-dependent and deterministic. We do the same approximation on~$X$ and get
\begin{equation}\label{SDEforX}dX_t = (\Omega+\epsilon^2(d-1) I^n_d + b_0^{A}(t)X_t +X_t (b_0^{A}(t))^{\top})dt+ \epsilon\left(\sqrt{X_t} dW^A_tI^n_d+I^n_d(dW^A_t)^{\top}\sqrt{X_t}\right),
\end{equation}

\noindent where \begin{equation*}
b_0^{A}(t)=b+\epsilon I^n_d\rho \sum_{k=1}^m\omega^k_0B(T+k\delta-t)^{\top}c+2\epsilon^2I^n_d\sum_{k=1}^m\omega^k_0D(T+k\delta-t). \label{def_bA0}
\end{equation*}
 Thanks to this approximation, we remark that the process, that we still denote by $(S_t,X_t)$ for simplicity, is now affine. This enables us to use again the same argument as for the Caplet prices to get an expansion of the price.  The only difference lies in the fact the expansion is around the Gaussian model rather then around the log-normal model.

\subsubsection{The swaption price expansion}

Let $P^S(t,x,s)=\E^{A}\left[\left(S_t-K\right)^+ |S_t=s,X_t=x \right]$ denote the price of the Swaption at time~$t\in [0,T]$. It solves the following pricing PDE
\begin{align*}
\partial_t P^S + \cL(t) P^S&= 0, \ t\in (0,T), \
P^S(T,x,s)=(s-K)^+,
\end{align*}
where $\cL$ is the infinitesimal generator of the SDE~\eqref{SAnnuityApprox} and~\eqref{SDEforX}.
Again, we assume that $P^S$ admits a second order expansion \begin{equation}\label{expansionPswaption}P^S=P^S_0+\epsilon P^S_1+\epsilon^2 P^S_2+o(\epsilon^2)
\end{equation}
and that the functions $P^S_0$, $P^S_1$ and $P^S_2$ are smooth enough. Let $\mathrm{BH}(s,v)=\E\left[\left(s+\sqrt{v}G-K\right)^+\right]$ with $G\sim N(0,1)$ denote the European call price with strike $K$ in the Bachelier model with realized volatility~$v>0$ and spot price $s\in \R$. We obtain
\begin{equation}
    \label{FormulaSwaptionExpOrder0}P^S_0(t,x,s)=\mathrm{BH}(s,v^S(t,T,x)),
\end{equation}

\noindent where

\begin{equation}\label{VolSwaption}
    v^S(t,T,x)=\int_t^{T}B^S(r)^{\top}c X^0_{r-t}(x) c^{\top}B^S(r)dr,
\end{equation}
and $X^0_s(x)$ is defined by~\eqref{defX0sx}. The higher order term are
\begin{eqnarray}
   \label{FormulaSwaptionExpOrder1} P^S_1(t,x,s)&=&\left(c^S_1(t,T,x)\partial^3_s+c^S_2(t,T,x)\partial^2_s\right)\mathrm{BH}(s,v^S(t,T,x)), \\
P^S_2(t,x,s)&=&\Bigg[d^S_1(t,T,x)\partial^4_s+d^S_2(t,T,x)\partial^3_s+d^S_3(t,T,x)\partial^2_s \\
    \label{FormulaSwaptionExpOrder2}&&+ e^S_1(t,T,x)\partial^6_s+e^S_2(t,T,x)\partial^5_s+e^S_3(t,T,x)\partial^4_s  \nonumber\\
&&  +e^S_4(t,T,x)\partial^4_s+e^S_5(t,T,x)\partial^3_s+e^S_6(t,T,x)\partial^2_s\Bigg]\mathrm{BH}(s,v^S(t,T,x)), \nonumber
\end{eqnarray}

\noindent where the coefficients $c^S_i$, $d^S_i$ and $e^S_i$ are given in Appendix~\ref{SwaptionExpansionCalcul}. Again, the derivatives of $P^S$ with respect to~$s$ can be calculated explicitly, which makes this formula very efficient from a computational point of view. %we can deduce the expansion of the implied volatility from this price expansion. 

\subsection{Numerical results}\label{Num_sec_exp}

We now assess on some examples the accuracy of the expansions we have developed. In practice we are interested in knowing up to what level of parameters and for what set of maturities and tenors the accuracy of the expansion is satisfactory. Let us recall that our expansion for caplets results from the combination of two expansions, the first on the support matrix function $D$ up to the order 1 in $\epsilon$ is given by (\ref{ExpSupportD0}) and (\ref{ExpSupportD1}), the second on the infinitesimal generator of the Markov process $(X,\varmaj)$ defined by (\ref{sigXL}). By construction the approximation of $D(\tau)$ will be more accurate for a small $\tau$. As a consequence, for a given set of parameters, the full expansion will likely to be more accurate for short maturities, short tenors caplets. The expansion for swaptions results from a supplementary approximation step, which consists in freezing the weights $\omega^i$ in the diffusion of the Markov process $(X,S)$ defined by (\ref{sigAnnuity}) and (\ref{SAnnuity}). This approximation can be inaccurate for long maturities and long tenors swaptions. Therefore, we expect the full expansion to be more accurate for short maturities, short tenors swaptions. \\ 

\begin{figure}[h!]
\centering
\includegraphics[scale = 0.35]{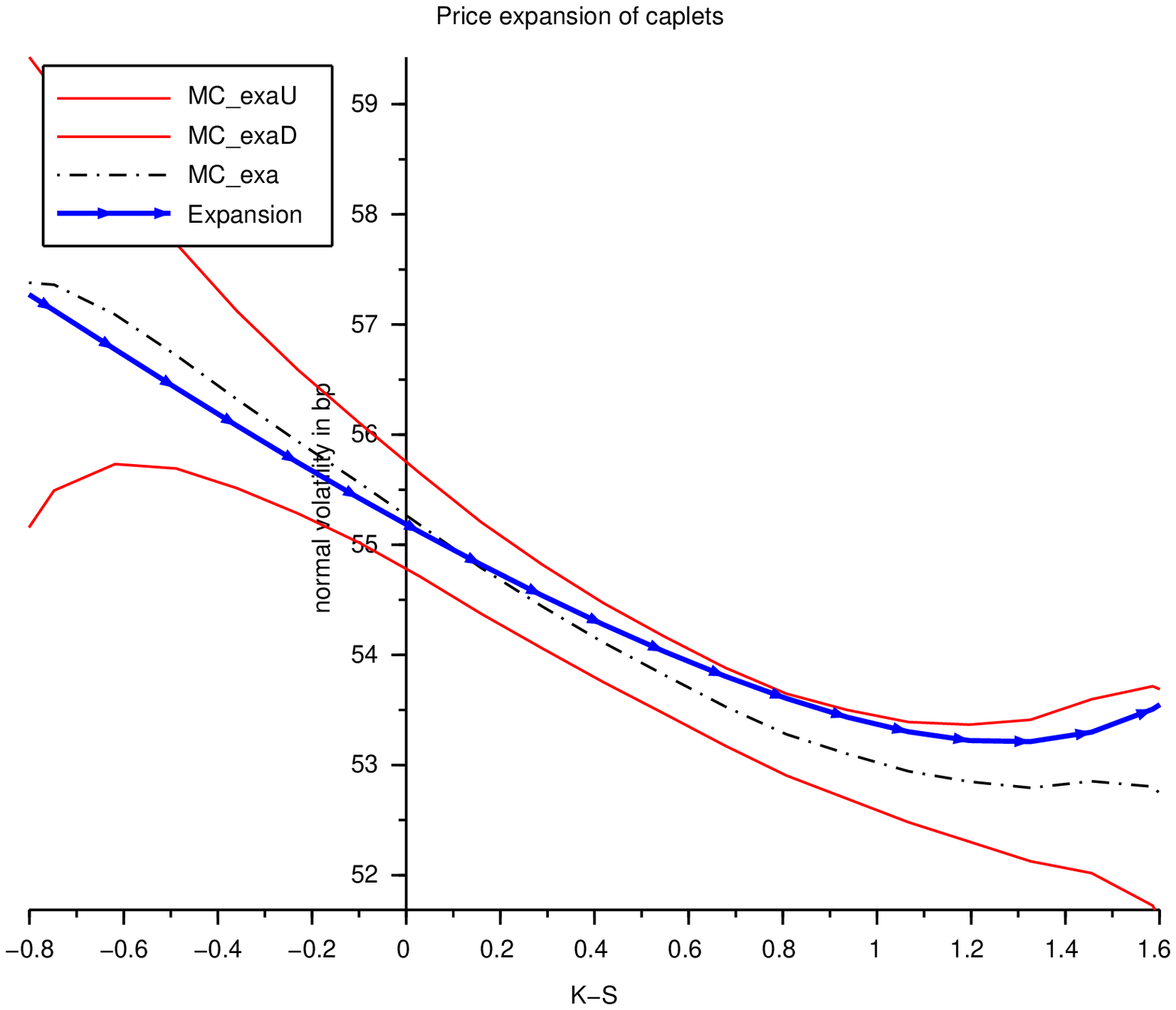}
\includegraphics[scale = 0.35]{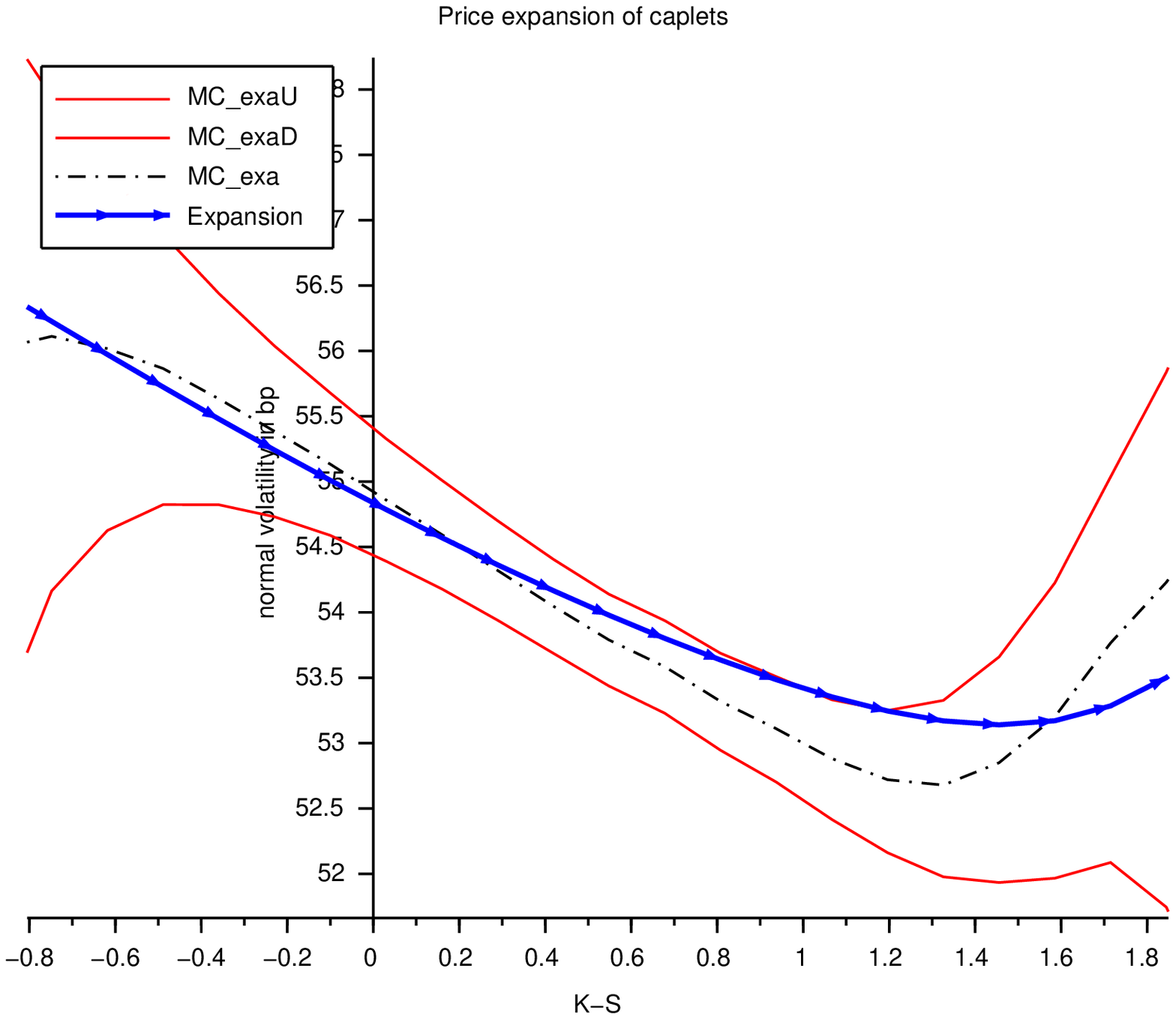}\\
\caption{$\rho^\top=(-0.4,-0.2)$. Plot of the expanded smile of a $1Y\times1Y$ caplet against the Monte Carlo smile obtained with 100000 paths and a discretization grid of 4 points for different values of the parameter $\epsilon$, respectively from left to right $\epsilon=0.002$ and $\epsilon=0.0015$. The forward Libor rate value is $L(0,1Y,1Y)=1.02\%$. }
\label{CapletSmileExp1}
\end{figure}

\begin{figure}[h!]
\centering
\includegraphics[scale = 0.35]{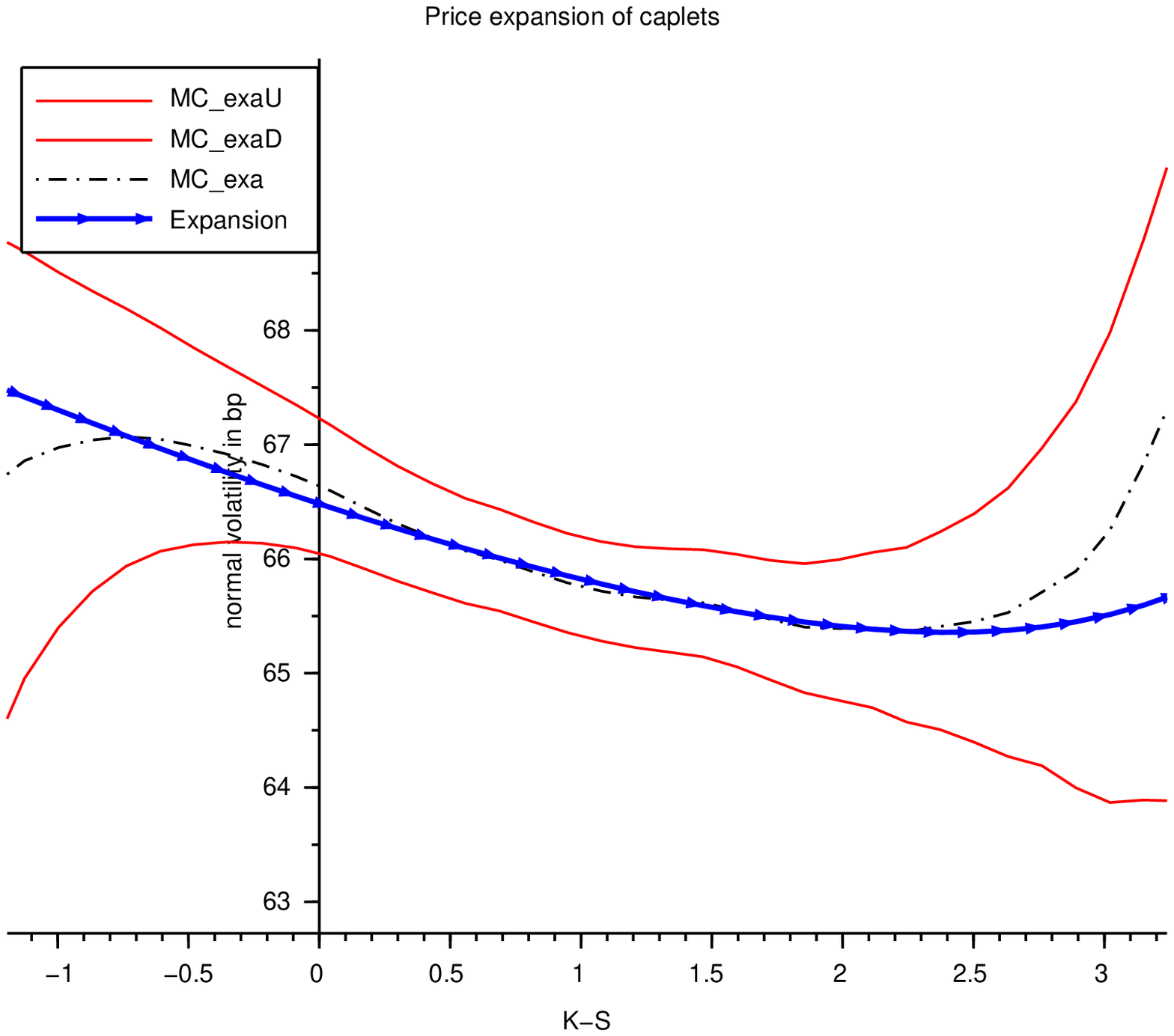}
\includegraphics[scale = 0.35]{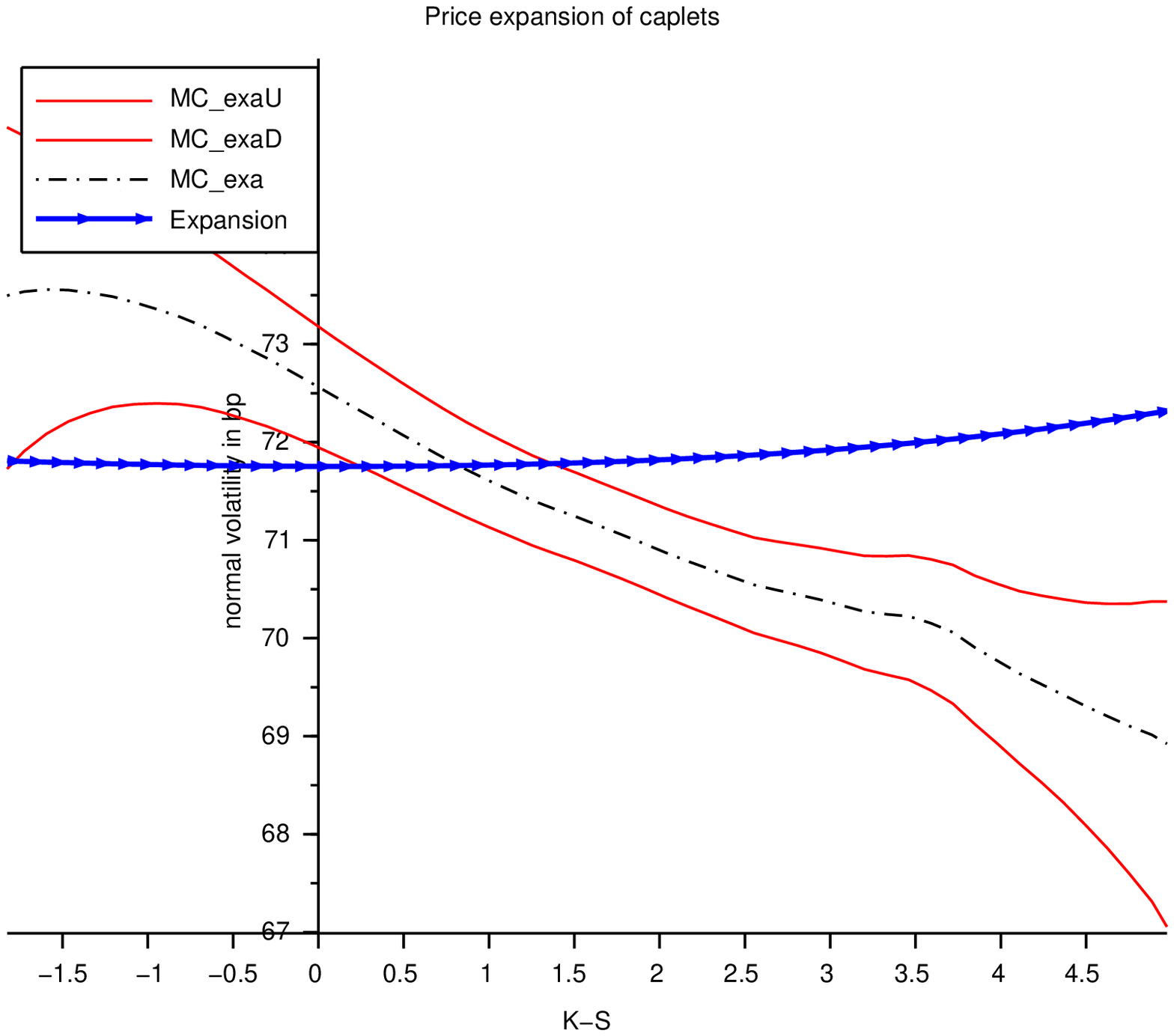}\\
\caption{$\rho^\top=(-0.4,-0.2)$ and $\epsilon=0.0015$. Left: plot of the expanded smile of a $6M\times2Y$ caplet against the Monte Carlo smile. Right: plot of the expanded smile of a $6M\times5Y$ caplet against the Monte Carlo smile. The Monte Carlo smile is obtained with 100000 paths and a discretization grid of 8 points. The forward Libor rates values are $L(0,6M,2Y)=1.14\%$ and $L(0,6M,5Y)=1.35\%$. }
\label{CapletSmileExp2}
\end{figure}

\begin{figure}[h!]
\centering
\includegraphics[scale = 0.35]{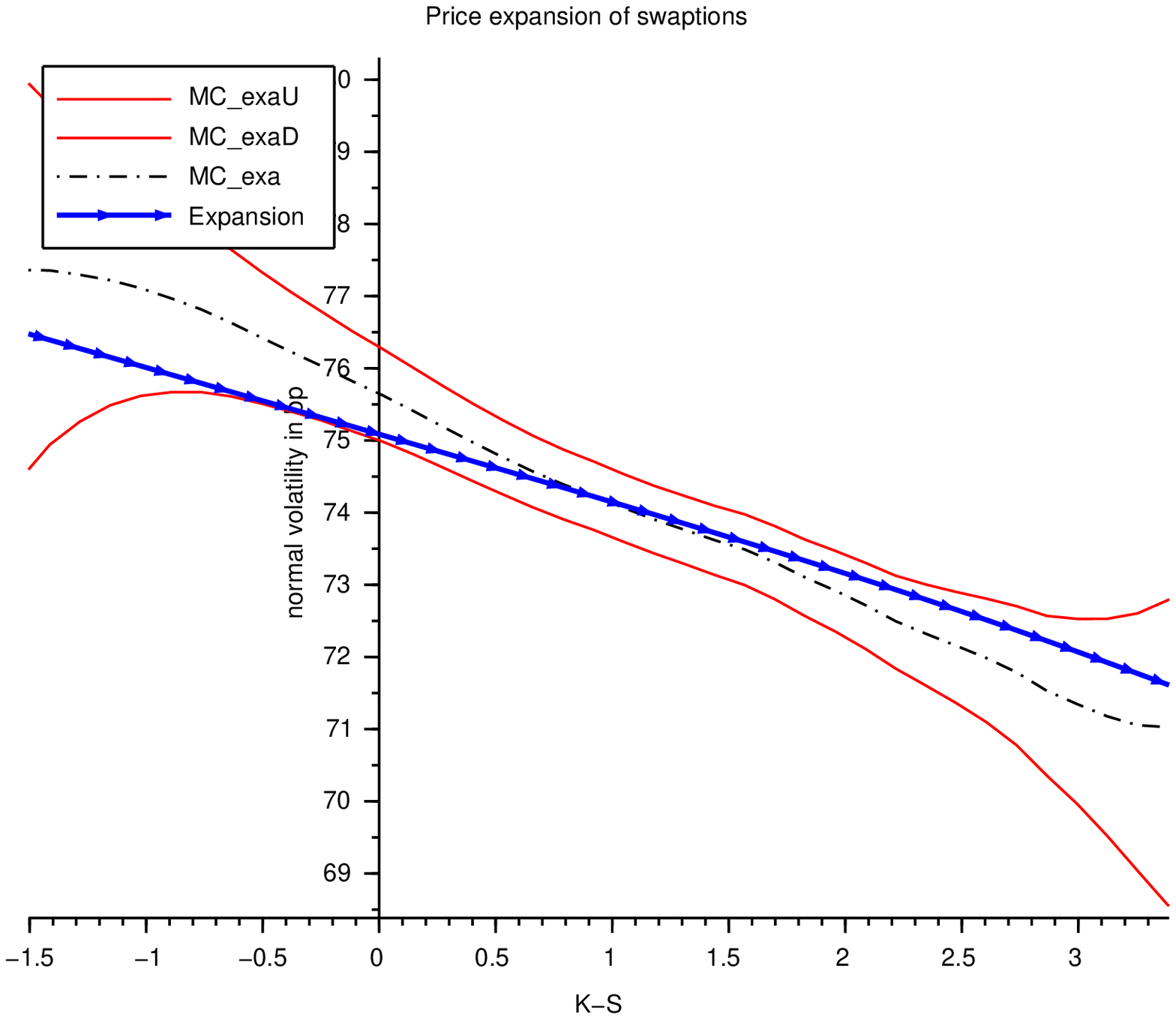}
\includegraphics[scale = 0.35]{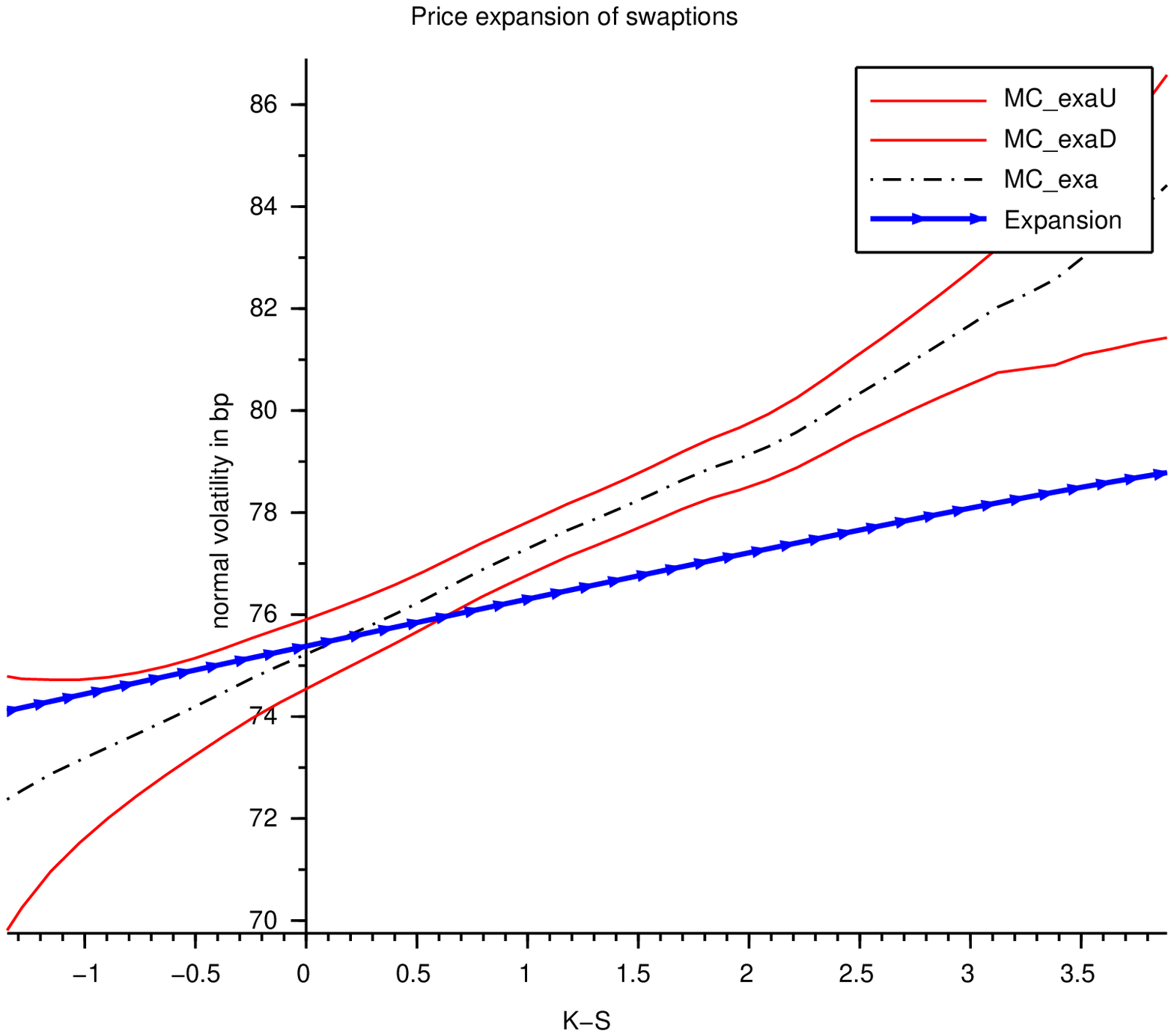}\\
\caption{$\epsilon=0.0015$. Plot of the expanded smile of a $5Y\times2Y$ swaption with coupon payment frequency of 6 months against the Monte Carlo smile obtained with 100000 paths and a discretization grid of 8 points for different values of the parameter $\rho$, from left to right $\rho^\top=(-0.4,-0.2)$ and $\rho^\top=(0.4,0.2)$. The forward swap rate value is $S(0,5Y,2Y)=1.3\%$.}
\label{SwaptionSmileExp3}
\end{figure}

We assess the quality of the price expansion for caplets and swaptions. We compare the expanded price with the price computed using Monte Carlo simulation and the discretization scheme~1 described in Section~\ref{Sec_scheme} on a regular time grid. The expanded prices and the Monte Carlo prices are compared in terms of the normal implied volatility of the forward Libor rate for caplets and of the forward swap rate for swaptions. The implied volatility is given in basis points ($10^{-4}$).  In abscissa is indicated the difference between the strike and the at-the-money value, and the unit is one percent. A $6M\times 2Y$ caplet will denote a caplet with maturity $T=2$ years and tenor $\delta=0.5$ years, while a $5Y\times 2Y$ swaption will denote a swaption with maturity $T=2$ years and tenor $m\delta=5$ years. \\

We have tested different sets of model parameters. The parameters values have been chosen in such a way that the yield curve and volatility levels generated by the model are in line with today's US and EUR interest rates market levels. Here, we only consider the following parameter set with $p=2$ and $d=2$:
\begin{align}
  \kappa &= \mathrm{diag}(0.1,1), \  c = I_d,\   b = -\mathrm{diag}(0.41,0.011), \ \Omega=-(bx_\infty+x_\infty b^{\top})+0.4I_d , \ \gamma=0.001I_d  \label{param_set} \\
  x &= 10^{-4}\left(
          \begin{array}{cccc}
            2.25 & -1.2\\
            -1.2 & 1.\\
          \end{array}
        \right), \ 
  x_\infty=10^{-4}\left(
          \begin{array}{cccc}
            1. & -0.125\\
            -0.125 & 0.25\\
          \end{array}
        \right).\nonumber
\end{align}
We note that $-(b+b^\top)=-2b$ is positive definite. We know from Remark~\ref{rk_cond_lap_2} that the condition of non-explosion will be verified  in general for these set of parameters when $\epsilon$ is small enough, and we have checked that the yield curve given by this parameter set is well defined up to 50 years.\\

\noindent In all the graphics the dotted line gives the Monte Carlo smile obtained with 100000 simulation paths, the solid line with small arrows is the expanded smile, the two continuous solid lines are the upper and lower bounds of the 95\% confidence interval of the Monte Carlo price. Figures~\ref{CapletSmileExp1} and \ref{CapletSmileExp2} show the accuracy of the expansion for the valuation of caplets.  The approximation is accurate for expiries up to 2 years and less accurate with the same parameters for longer expiries. For maturities up to 2 years, the at-the-money volatility of the expanded smile is almost identical to the Monte Carlo smile and the whole expanded smile stays within the 95\% confidence interval. Figure~\ref{SwaptionSmileExp3} shows the accuracy of the expansion for the valuation of swaptions. We observe that the expansion is more accurate for negative values of the correlation parameters $\rho$ (a similar behaviour is observed for Caplets). This can be intuitively understood from the Riccati equation~\eqref{LaplaceGenRiccati}: a negative $\rho$ pushes $D$ to zero while a positive one pushes $D$ away from zero, and the expansion that we use on~$D$ (see~(\ref{ExpSupportD0}) and (\ref{ExpSupportD1})) is then less accurate.  Overall the expansion is accurate at-the-money and is much less accurate out-of-the-money. For example, the graphic on the right hand side of Figure~\ref{CapletSmileExp2} shows that the expanded smile of the 6 months maturity 5 years expiry smile is quite inaccurate and the expanded smile fails to fit the skew of the Monte Carlo smile. However, the difference in the at-the-money volatility between the expanded price and Monte Carlo is around 1 bp.\\

To sum up, the second order expansion is basically accurate for small perturbations and small maturities. Otherwise, one should be careful and rely on other methods such as the Monte-Carlo method or Fourier inversion method. Nonetheless, as discussed in Section~\ref{sec_comp_meth}, the calculation of the expansion is much faster than the other methods. It may thus be relevant to start a calibration routine and select a reasonable set of parameters.

%***********************************************************************************************************************************************************
%***********************************************************************************************************************************************************

\section{Second order discretization schemes for Monte Carlo simulation}\label{Sec_scheme}

The goal of this section is to construct discretization schemes for the process $(X,Y)$ defined by~\eqref{SDE_GeneralY} and~\eqref{SDE_General}. It is crucial to have an efficient way to simulate the model in order to use it in practice. Ideally, the model should be calibrated to market data to vanilla options such as caplets and swaptions and then be used to calculate exotic option prices. The calculation of these prices is generally made with a Monte-Carlo algorithm which requires to simulate the process~$(X,Y)$.

It is worth to recall that the standard Euler-Maruyama scheme is not well defined for square-root diffusions even in dimension one, see Alfonsi~\cite{Alfonsi1}. We have then to consider a different scheme. We use here the splitting technique that is already used by Ahdida and Alfonsi~\cite{Alfonsi2} for Wishart processes. We explain here briefly the main line of this method and refer to~\cite{Alfonsi1} for precise statements in a framework that embeds affine diffusions. Let us consider that we want to approximate an SDE~$\xi$ with infinitesimal generator~$\cL$ on the regular time grid $t_i=iT/N$, for $i=0,\dots,N$. A scheme is fully described by a probability law $\hat{p}_x(t,dz)$ that approximates the law of $\xi_t$ given $\xi_0=x$. We denote by $\hat{\xi}^x_t$ a random variable following this law. Then, the law of the corresponding discretization scheme $(\hat{\xi}_{t_i},0\le i\le N)$ is as follows: $\hat{\xi}_0=\xi_0$ and $\hat{p}_{\hat{\xi}_{t_i}}(T/N,dz)$ is the conditional law of $\hat{\xi}_{t_{i+1}}$ given $(\hat{\xi}_{t_j},0\le j\le i)$. Then, one would like to know the error made when using the approximation scheme instead of the original process~$\xi$. We have basically the following result, up to technical details that are given in~\cite{Alfonsi1}. If $\hat{\xi}^x_t$  satisfies the following expansion
\begin{equation}\label{exp_sch_O2}\E[f(\hat{\xi}^x_t)]=f(x)+t\cL f(x)+\frac{t^2}{2}\cL^2f(x)+O(t^3)
\end{equation}
for any smooth function~$f$, then
$$\exists C>0, \ |\E[f(\hat{\xi}_{t_N})]- \E[f({\xi}_{t_N})]|\le C/N^2.$$
Thus, to get a weak error of order~$2$, we mainly have to construct a scheme $\hat{\xi}^x_t$ that satisfies~\eqref{exp_sch_O2}. We can construct iteratively second order schemes by splitting the infinitesimal generator. In fact, let us assume that $\cL=\cL_1+\cL_2$ and that $\hat{\xi}^{i,x}_t$ is a second order scheme for~$\cL_i$. Let $\mathbf{B}$ be an independent Bernoulli variable with parameter~$1/2$. Then, the following schemes
\begin{equation}\label{compo_rule}\hat{\xi}^{1,\hat{\xi}^{2,\hat{\xi}^{1,x}_{t/2} }_t}_{t/2} \text{ and } \mathbf{B}\hat{\xi}^{2,\hat{\xi}^{1,x}_t}_t + (1-\mathbf{B})\hat{\xi}^{1,\hat{\xi}^{2,x}_t}_t
\end{equation}
satisfy~\eqref{exp_sch_O2} and are thus second order schemes for~$\cL$. Therefore, a strategy to construct a second order scheme is to split the infinitesimal generator into elementary pieces for which second order schemes or even exact schemes are known.

To use this splitting technique, we first have to calculate the infinitesimal generator of~$(X,Y)$. It is defined for a $\mathcal{C}^2$ function $f: \genm \times \R^p \rightarrow \R$ by $\cL f(x,y)=\lim_{t\rightarrow 0^+} \frac{\E[f(X_t,Y_t)]-f(x,y)}{t}$. From~\eqref{crochets_1},~\eqref{crochets_2} and~\eqref{crochets_3}, we easily get
\begin{align*}
\cL =& \sum_{m=1}^p (\kappa(\theta-y))_m\partial_{y_m}+ \sum_{1\le i,j\le d}(\Omega+(d-1)\epsilon^2I^n_d+bx+xb^\top)_{i,j} \partial_{x_{i,j}}\\
&+\frac 12 \sum_{m,m'=1}^p (cxc^\top)_{m,m'}\partial_{y_m}\partial_{y_{m'}}+ \frac 12 \sum_{m=1}^p\sum_{1\le i,j\le d} \epsilon [(cx)_{m,i}(I^n_d \rho)_j +(cx)_{m,j}(I^n_d \rho)_i ] \partial_{x_{i,j}}\partial_{y_m}\\&+ \frac 12 \sum_{1\le i,j,k,l \le d} \epsilon^2 [x_{i,k}  (I^n_d)_{j,l} +x_{i,l} (I^n_d)_{j,k}+x_{j,k} (I^n_d)_{i,l} +x_{j,l}(I^n_d)_{i,k} ]\partial_{x_{i,j}}\partial_{x_{k,l}}.
\end{align*}
Here, $\partial_{y_m}$ denotes the partial derivative with respect to the $m$-th coordinate in $\R^p$ and $\partial_{x_{i,j}}$ the partial derivative with respect to the element at the $i$-th row and $j$-th column. When $\rho=0$, this operator is simply the sum of the infinitesimal generators for~$X$ and the generator for~$Y$ when $X$ is frozen. We know from~\cite{Alfonsi2} a second order scheme for~$X$. When $X$ is frozen, $Y$ follows an Ornstein-Uhlenbeck process and the law of $Y_t$ is a Gaussian vector that can be sampled exactly. By using the composition rule~\eqref{compo_rule}, we get a second order scheme for $(X,Y)$.

Thus, the difficulty here comes from the correlation between $X$ and $Y$ that has to be handled with care. We first make some simplifications.
The first term $ \sum_{m=1}^p (\kappa(\theta-y))_m\partial_{y_m}$ is the generator of the linear Ordinary Differential Equation $y'(t)=\kappa(\theta -y(t))$ that is solved exactly by $y(t)=e^{-\kappa t}y(0)+(I_p-e^{-\kappa t})\theta$. Therefore, it is sufficient to have a second order scheme for $\cL-\sum_{m=1}^p (\kappa(\theta-y))_m\partial_{y_m}$, which is the generator of~\eqref{SDE_GeneralY} and~\eqref{SDE_General} when $\kappa=0$. When $\kappa=0$, we have $Y_t=y+c(\tilde{Y_t}-\tilde{Y_0})$ with
$$\tilde{Y_t}=\tilde{Y_0}+ \int_0^t \sqrt{X_s}\left[\bar\rho dZ_s + dW_s\rho \right].$$
We can then focus on getting a second order scheme for $(X,\tilde{Y})$, which amounts to work with $p=d$ and $c=I_d$. It is therefore sufficient to find a second order scheme for the SDE
\begin{align*}
 Y_t &=y + \int_0^t \sqrt{X_s} \left[ \bar\rho dZ_s + dW_s \rho \right],\\
X_t &=x +\int_0^t\left( \Omega +(d-1)\epsilon^2I^n_d  + bX_s+X_sb^{\top} ) \right )ds + \epsilon\int_0^t\sqrt{X_s}dW_sI^n_d + I^n_ddW_s^{\top} \sqrt{X_s},
\end{align*}
with the infinitesimal generator
\begin{align}
\cL =&   \sum_{1\le i,j\le d}(\Omega+(d-1)\epsilon^2I^n_d+bx+xb^\top)_{i,j} \partial_{x_{i,j}}+ \frac 12 \sum_{m=1}^d\sum_{1\le i,j\le d} \epsilon [x_{m,i}(I^n_d \rho)_j +x_{m,j}(I^n_d \rho)_i ] \partial_{x_{i,j}}\partial_{y_m}  \label{IG_bis}\\
&+\frac 12 \sum_{m,m'=1}^d x_{m,m'}\partial_{y_m}\partial_{y_{m'}} + \frac 12 \sum_{1\le i,j,k,l \le d} \epsilon^2 [x_{i,k}  (I^n_d)_{j,l} +x_{i,l} (I^n_d)_{j,k}+x_{j,k} (I^n_d)_{i,l} +x_{j,l}(I^n_d)_{i,k} ]\partial_{x_{i,j}}\partial_{x_{k,l}}. \nonumber
\end{align}

\subsection{A second order scheme}\label{2nd_order_gen}

For $1\le q \le d$, we define $e^q_d\in \posm$ by $(e^q_d)_{k,l}=\indi{k=l=q}$  and $g^q_d \in \R^d$ by $(g^q_d)_k=\indi{q=k}$ so that $I^n_d=\sum_{q=1}^n e^q_d$ and $I^n_d \rho=\sum_{q=1}^n  \rho_q g^q_d$. We define
\begin{align}
\cL^c_q=&  \epsilon^2(d-1)\partial_{x_{q,q}} + \frac 12 \sum_{m=1}^d\sum_{1\le i,j\le d} \epsilon \rho_q [x_{m,i}(g^q_d)_j +x_{m,j}(g^q_d)_i ] \partial_{x_{i,j}}\partial_{y_m}  \label{def_Lq}+\frac {\rho_q^2}2 \sum_{m,m'=1}^d x_{m,m'}\partial_{y_m}\partial_{y_{m'}}\\& + \frac 12 \sum_{1\le i,j,k,l \le d} \epsilon^2 [x_{i,k}  (e^q_d)_{j,l} +x_{i,l} (e^q_d)_{j,k}+x_{j,k} (e^q_d)_{i,l} +x_{j,l}(e^q_d)_{i,k} ]\partial_{x_{i,j}}\partial_{x_{k,l}}.\nonumber
\end{align}
We consider the splitting $\cL=\cL'+\cL''+\sum_{q=1}^n \cL^c_q$ of the operator~\eqref{IG_bis}, with
\begin{align*}
\cL'=& \sum_{1\le i,j\le d}(\Omega+bx+xb^\top)_{i,j} \partial_{x_{i,j}}, \\
\cL''=&  \left(1-\sum_{q=1}^n \rho_q^2 \right) \frac 12 \sum_{m,m'=1}^d x_{m,m'}\partial_{y_m}\partial_{y_{m'}}. \\
\end{align*}
The operator $\cL'$ is the one of the linear ODE $x'(t)=\Omega+(d-1)\epsilon^2I^n_d+bx+xb^\top$ that can be solved exactly and stays in the set of semidefinite positive matrices, see Lemma~27 in~\cite{Alfonsi2}. The operator $\cL''$ is the one of $Y''_t=y''+\sqrt{1-\sum_{q=1}^n \rho_q^2} \sqrt{x} Z_t$, which can be sampled exactly since it is a Gaussian vector with mean $y''$ and covariance matrix $(1-|\rho|^2) t x $. The operator $\cL^c_q$ is the infinitesimal generator of the following SDE
\begin{eqnarray}
\begin{cases} Y_t &= y +\rho_q\int_0^t\sqrt{X_s} dW_sg^q_d,  \\
X_t &=x+\int_0^t(d-1) \epsilon^2e^q_d  ds + \epsilon \int_0^t\sqrt{X_s}dW_se^q_d + e^q_ddW_s^{\top}\sqrt{X_s}.
\end{cases}
\label{diff_Lcq}
\end{eqnarray}
Thus, $X$ follows an elementary Wishart process and stays in~$\posm$. Using the notation of~\cite{Alfonsi2}, $X_t$ follows the law $WIS_d(x,d-1,0,e^q_d,\epsilon^2 t)$. Theorems~9 and~16 in~\cite{Alfonsi2} gives respectively an exact and a second (or higher) discretization scheme for this process. We now explain how to calculate $Y_t$ once that $X_t$ has been sampled. From~\eqref{diff_Lcq}, we have for $1\le i\le d$,
\begin{align*}
d(Y_t)_i=&\rho_q \sum_{j=1}^d  (\sqrt{X_t})_{i,j}(dW_t)_{j,q}, \\
d(X_t)_{q,i}=& \epsilon \sum_{j=1}^d  (\sqrt{X_t})_{i,j}(dW_t)_{j,q} +\indi{i=q} \left[(d-1)\epsilon^2dt +  \sum_{j=1}^d  (\sqrt{X_t})_{q,j}(dW_t)_{j,q} \right]. \\
\end{align*}
This yields to
\begin{align*}
(Y_t)_i&=y_i+ \frac{\rho_q}{\epsilon}((X_t)_{q,i}-x_{q,i}), \text{ if }i\not = q,\\
(Y_t)_q&=y_i+ \frac{\rho_q}{2\epsilon}[(X_t)_{q,q}-x_{q,q}-\epsilon^2(d-1)t].
\end{align*}
Using these formula together with the exact (resp. second order) scheme for~$X_t$, we get an exact (resp. second order) scheme for~\eqref{diff_Lcq}. By using the composition rules~\eqref{compo_rule}, we get a second order scheme for~\eqref{IG_bis}.

\subsection{A faster second order scheme when $\Omega-\epsilon^2I^n_d \in \posm$}\label{2nd_order_fast}

As explained in~\cite{Alfonsi2}, the sampling of each elementary Wishart process in~$\cL_q$ requires a Cholesky decomposition that has a time complexity of $O(d^3)$. Since the second order scheme proposed above calls $n\le d$ times this routine, the whole scheme requires at most $O(d^4)$ operations. However, by adapting an idea that has been already used in~\cite{Alfonsi2} for Wishart processes, it is possible to get a faster scheme if we assume in addition that  $\Omega-\epsilon^2I^n_d \in \posm$. We now present this alternative scheme that only requires $O(d^3)$ operations.

We consider the splitting $\cL=\tilde{\cL}'+\tilde{\cL}''+ \hat{\cL}$ of the operator~\eqref{IG_bis}, with
\begin{align*}
\tilde{\cL}'=& \sum_{1\le i,j\le d}(\Omega-\epsilon^2I^n_d+bx+xb^\top)_{i,j} \partial_{x_{i,j}} \\
%\tilde{\cL}''=&  \left(1-\sum_{q=1}^d \rho_q^2 \right) \frac 12 \sum_{m,m'=1}^d x_{m,m'}\partial_{y_m}\partial_{y_{m'}}\\
\hat{\cL}=&\sum_{1\le i\le n}  d \epsilon^2 \partial_{x_{i,i}}+ \frac 12 \sum_{m=1}^d\sum_{1\le i,j\le d} \epsilon [x_{m,i}(I^n_d \rho)_j +x_{m,j}(I^n_d \rho)_i ] \partial_{x_{i,j}}\partial_{y_m} +  \frac {\sum_{q=1}^n \rho_q^2}2 \sum_{m,m'=1}^d x_{m,m'}\partial_{y_m}\partial_{y_{m'}} \label{IG_bis}\\
& + \frac 12 \sum_{1\le i,j,k,l \le d} \epsilon^2 [x_{i,k}  (I^n_d)_{j,l} +x_{i,l} (I^n_d)_{j,k}+x_{j,k} (I^n_d)_{i,l} +x_{j,l}(I^n_d)_{i,k} ]\partial_{x_{i,j}}\partial_{x_{k,l}}.
\end{align*}
Again,  $\tilde{\cL}'$ is the operator of the linear ODE $x'(t)=\Omega-\epsilon^2I^n_d+(d-1)\epsilon^2I^n_d+bx+xb^\top$ that can be solved exactly and stays in the set of semidefinite positive matrices by Lemma~27 in~\cite{Alfonsi2} since $\Omega-\epsilon^2I^n_d \in \posm$. We have already seen above that the generator $\cL''$  can be sampled exactly, and we focus now on the sampling of~$\hat{\cL}$. It relies on the following result.
\begin{lemma}\label{Lemma_fast_sch} For $x \in \posm$ we consider $c\in \genm$ such that $c^\top c=x$. We define $U_t=c+\epsilon W_tI^n_d$, $X_t=U_t^\top U_t$ and $Y_t=y+\int_0^t U_s^\top dW_s I^n_d \rho$. Then, the process $(X,Y)$  has the infinitesimal generator~$\hat{\cL}$.
\end{lemma}
\begin{proof}For $1\le i,j,m\le d$, we have
$d(X_t)_{i,j}=   \epsilon  \sum_{k=1}^d \left((U_t)_{k,i} (dW_t)_{k,j}\indi{j \le n} + (U_t)_{k,j} (dW_t)_{k,i}\indi{i \le n}\right)+ \indi{i=j \le n} d \epsilon^2  dt$ and $d(Y_t)_m=\sum_{k,l=1}^d (U_t)_{k,m} (dW_t)_{k,l}(I^n_d\rho)_l$.
This leads to
\begin{align*}
\langle d(Y_t)_m, d(Y_t)_{m'} \rangle=& \sum_{k,l=1}^d (U_t)_{k,m}(U_t)_{k,m'} (I^n_d\rho)_l^2 dt= \left(\sum_{l=1}^n \rho_l^2  \right)(X_t)_{m,m'}dt,\\
\langle d(Y_t)_m, d(X_t)_{i,j} \rangle=&  \epsilon [(I^n_d \rho)_j (X_t)_{m,i}+ (I^n_d\rho)_i (X_t)_{m,j}]dt,\\
\langle d(X_t)_{i,j}, d(X_t)_{k,l} \rangle=&\epsilon^2 [(X_t)_{i,k}  (I^n_d)_{j,l} +(X_t)_{i,l} (I^n_d)_{j,k}+(X_t)_{j,k} (I^n_d)_{i,l} +(X_t)_{j,l}(I^n_d)_{i,k} ] dt,
\end{align*}
which precisely gives the generator~$\hat{\cL}$.
\end{proof}

Thanks to Lemma~\ref{Lemma_fast_sch}, it is sufficient to construct a second order scheme for $(U,Y)$. Since $\langle d(Y_t)_m, d(U_t)_{i,j} \rangle = \epsilon (U_t)_{i,m} (I^n_d\rho)_jdt$, the infinitesimal generator~$\bar{\cL}$ of $(U,Y)$ is given by
$$ \bar{\cL}=\frac{\epsilon^2}{2} \sum_{i=1}^d\sum_{j=1}^n\partial^2_{x_{i,j}}+\frac{\epsilon}{2}\sum_{i,m=1}^d\sum_{j=1}^n \rho_j x_{i,m}\partial_{x_{i,j}}\partial_{y_m}+\frac {\sum_{q=1}^n \rho_q^2}{2}\sum_{m,m'=1}^d(x^\top x)_{m,m'}\partial_{y_m}\partial_{y_m'}.$$
We use now the splitting $\bar{\cL}=\sum_{q=1}^n \bar{\cL}_q$ with
$$\bar{\cL}_q=\frac{\epsilon^2}{2} \sum_{i=1}^d \partial^2_{x_{i,q}}+\frac{\epsilon}{2}\sum_{i,m=1}^d  \rho_q x_{i,m}\partial_{x_{i,q}}\partial_{y_m}+\frac {  \rho_q^2}{2}\sum_{m,m'=1}^d(x^\top x)_{m,m'}\partial_{y_m}\partial_{y_m'}. $$
By straightforward calculus, we find that $\bar{\cL}_q$ is the generator of the following SDE
\begin{equation*}
d Y_t =  \rho_qU_t^{\top} dW_tg_d^q,\,\, dU_t =\epsilon dW_te_d^q.
\end{equation*}
We note that only the $q^{\textup{th}}$ row of~$U$ is modified. For $1\le i\le d$ we have
$d(U_t)_{i,q}=\epsilon (dW_t)_{i,q}$ and $d(Y_t)_m= \rho_q \sum_{j=1}^d (U_t)_{j,m}(dW_t)_{j,q}$. This yields to
\begin{align*}
(Y_t)_m&= (Y_0)_m+ \rho_q \sum_{j=1}^d  (U_0)_{j,m}(W_t)_{j,q} \text{ for } m \not = q,\\
(Y_t)_q&= (Y_0)_q+\rho_q \sum_{j=1}^d  (U_0)_{j,q}(W_t)_{j,q} + \frac{\epsilon \rho_q}{2} \sum_{j=1}^d \{ (W_t)_{j,q}^2-t \}.
\end{align*}
By using these formulas, we can then sample exactly $(U_t,Y_t)$ and then get a second order scheme for~$\hat{\cL}$. We note that the simulation cost of $\bar{\cL}_q$ requires $O(d)$ operations and then the one of~$\bar{\cL}$ requires $O(d^2)$ operations. Since a matrix multiplication requires $O(d^3)$ operations, this second order scheme for~$\bar{\cL}$ and then for $\cL$ requires $O(d^3)$ operations instead of $O(d^4)$ for the scheme described in Subsection~\ref{2nd_order_gen}.

\begin{remark}
As already mentioned, the dependence between the processes $X$ and $Y$ is the same as the one proposed by Da Fonseca, Grasselli and Tebaldi~\cite{F} for a model on asset returns. Therefore, we can use the same splittings as the one proposed in Subsections~\ref{2nd_order_gen} and~\ref{2nd_order_fast} to construct second order schemes for their model.
\end{remark}

\subsection{Numerical results}
\begin{figure}[h!]
\centering
\includegraphics[height = 4.5cm]{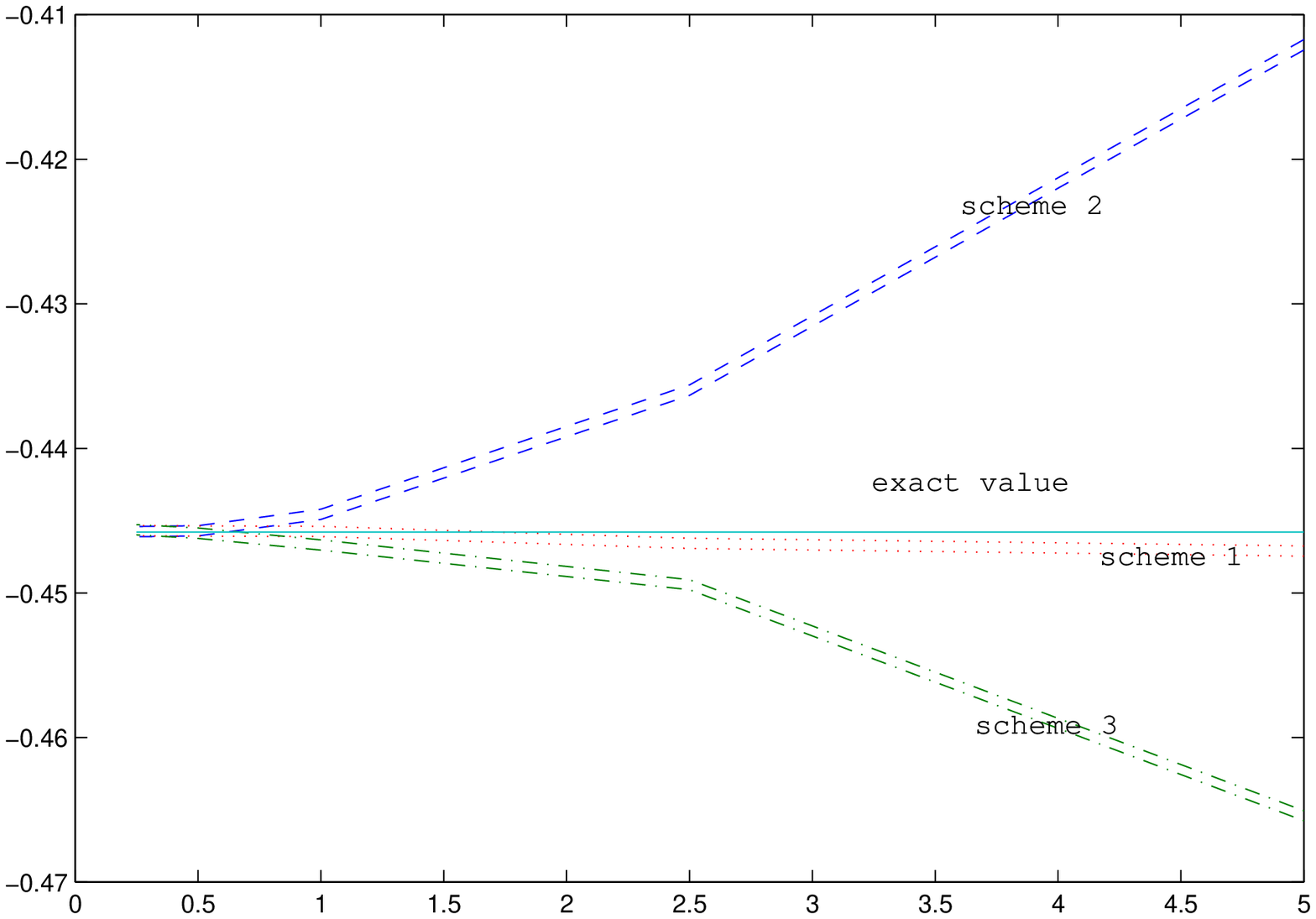}
\hspace{-0.7cm}
\includegraphics[height = 4.5cm]{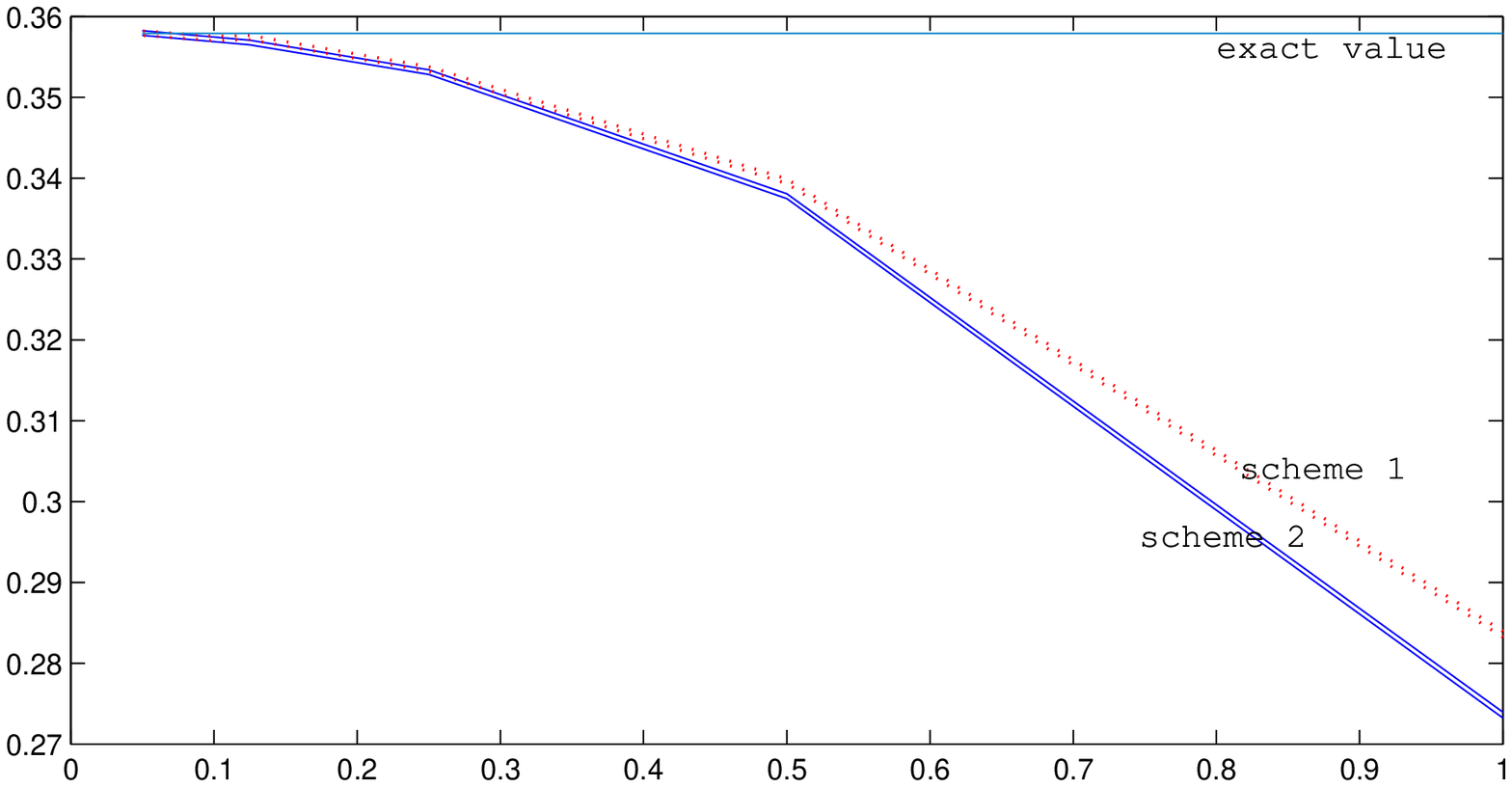}
\caption{Weak error convergence. Parameters: $p=d=3$, $10^7$ Monte Carlo samples, $T=5$. The real value of $\E\left[\exp\left(-i\left(\Tr(\Gamma X_T)+\Lambda^{\top}Y_T\right)\right)\right]$, as a function of the time step $T/N$. Left: $\Gamma=0.05 I_d, \Lambda=0.02\mathbf{1}_d$ and the diffusion parameters $x=0.4I_d, y=0.2\mathbf{1}_d, \Omega=2.5I_d, n=d, \rho=0, b=0, \kappa=0, c=I_d$. The value obtained by solving the ODE: $-0.445787$. Right: $\Gamma=0.2I_d+0.04q, \Lambda=0.2\mathbf{1}_d$ and the diffusion parameters $x=0.4I_d+0.2q, y=0.2 \mathbf{1}_d, \Omega=0.5I_d, n=d, \rho=-0.3\mathbf{1}_d, b=-0.5I_d, \kappa=0.1I_p, c=I_d$, where $q_{i,j}=\indi{i\neq j}$. The value obtained by solving the ODE: $0.357901$. 
For each scheme, the two curves represent the upper and lower bound of the 95\% confidence interval.}
\label{SchemeCV}
\end{figure}

We now turn to the empirical analysis of the convergence of the discretization schemes we have proposed. We will use the following notations.
\begin{itemize}
  \item Scheme 1 is the second order scheme given in Subsection~\ref{2nd_order_gen}, where we use the exact sample of the Wishart part and the exact simulation the Gaussian variables.
  \item Scheme 2 is the second order scheme given in Subsection~\ref{2nd_order_gen}, where  we use the second order scheme for the Wishart part  and replacing the simulation of Gaussian variables by random variables that matches the five first moments, see Theorem~16 and equation~(36) in~\cite{Alfonsi2}.
  \item Scheme 3 is the second order scheme given in Subsection~\ref{2nd_order_fast}.
\end{itemize}

\noindent In order to assess that the potential second order schemes we have proposed for $\cL$ give indeed a weak error of order~$2$, we start by analyzing the weak error for quantities that we can compute analytically. Namely, we consider $\E\left[\exp\left(-i\left(\Tr(\Gamma X_T)+\Lambda^{\top}Y_T\right)\right)\right]$, which can be calculated by solving a system of differential equations similar to~\eqref{LaplaceGenRiccati}. We then compare the values obtained by Monte Carlo simulation and the value obtained by solving the system of differential equation.
As shown by Figures~\ref{SchemeCV}, we observe a weak error which is compatible with the rate of  $O(1/N^2)$. When it is well defined, Scheme 3 has to be preferred since it is much faster than the others.

\section{Comparison of the different numerical methods}\label{sec_comp_meth}

The goal of this section is to compare the computational time needed to price vanilla instruments in the model by using the different numerical methods. We consider the case of a 6M$\times$1Y caplet with strike 1\% , which means $T=1$, $\delta=1/2$ and its price is given by $$\frac{1}{\delta}\E[e^{-\int_0^{T}r_sds}\left(1-(1+K\delta)P_{T,T+\delta} \right)^+]=\frac{P_{0,T}}{\delta} \E^T[\left(1-(1+K\delta)P_{T,T+\delta} \right)^+].$$ 
We will compare the expansion and the Monte-Carlo method with respect to the Fourier inversion method presented by Carr and Madan~\cite{CM} and Lee~\cite{L}. Their approach can be directly applied  for Caplets by working with the forward Caplet price. Let us note that this method can be adapted for swaptions by making the same approximation as the one that we use for the expansion, see Schrager and Pelsser~\cite{citeulike:824660} and Singleton and Umantsev~\cite{MAFI:MAFI427}. 
We consider here the four following numerical methods.
\begin{itemize}
\item The Monte-Carlo method that consists in using the second order scheme for $(X,Y)$ with a time step of $1/8$ and $10000$ paths in order to approximate $\frac{1}{\delta}\E[e^{-\int_0^{T}r_sds}\left(1-(1+K\delta)P_{T,T+\delta} \right)^+]$.
\item The expansion up to order~$2$. The integrals that define the coefficients $c_i$, $d_i$ and $e_i$ are approximated by using a trapezoidal rule and a time step of $1/20$. 
\item The Fourier transform under $\Px^T$. Starting from the expectation under the $T$-forward measure, we use the construction of Carr and Madan~\cite{CM}. In equation~(5) of~\cite{CM}, we use $\alpha=1.25$, truncate the integral at $375$ and use a Simpson's rule with a discretization step of $1/8$. Since we calculate here only one price, we do not use the FFT which would have generated further constraints between the discretization and strike grids. 
\item The Fourier transform under $\Px^{T+\delta}$. This is the same method starting with formula~\eqref{CallForward}, and we use the same parameters to approximate the integral and for $\alpha$. 
\end{itemize}

\begin {table}[h!]
\begin{center}
\begin{tabular}{|c|l|l|}
  \hline
  Pricing Method & Price (bp)  & Cpu time (s)\\
  \hline
  MC price &51.75 $\pm$ 1.46  (95\%  CI) & 43.3 \\
  Expansion & 52.33 & 0.686 \\
  Fourier under $P^T$ & 53.84  & 31.6 \\
  Fourier under $P^{T+\delta}$ & 52.87  & 33.6 \\
  \hline
\end{tabular}
\caption{Price of the 6M$\times$1Y caplet with strike 1\% using different methods with parameter set~\eqref{param_set} and $\rho=(-0.4,-0.2)$. Computations are made on a personal laptop with 4Go RAM and a 2.13GHz CPU.}
\end{center}
\end {table}
\noindent The striking fact is that the method based on the Fourier transform is not so efficient in this context, even though the Fourier inversion is in dimension one. The reason is that the evaluation of the Fourier transform requires to solve numerically matrix Riccati differential equations, for which we take a time step of $1/8$. Figure~\ref{FourierCv2} indicates on our case that a minimum of 2000 evaluations is necessary to have a precision similar to the Monte-Carlo method. Thus, a basic application of the method of Carr and Madan is not very efficient: the bottleneck is to find a smarter way to calculate the characteristic function. In comparison, the Monte Carlo method is not much more time consuming and allows to calculate the price for all strikes and maturities at the same time. Last, we observe that the expansion method is much faster than the others, but is limited to short maturities as indicated in Subsection~\ref{Num_sec_exp}. It can therefore be a tool to calibrate quickly the model to some key features such as the at the money price and skew.

\begin{figure}[h!]
\centering
\includegraphics[scale = 0.4]{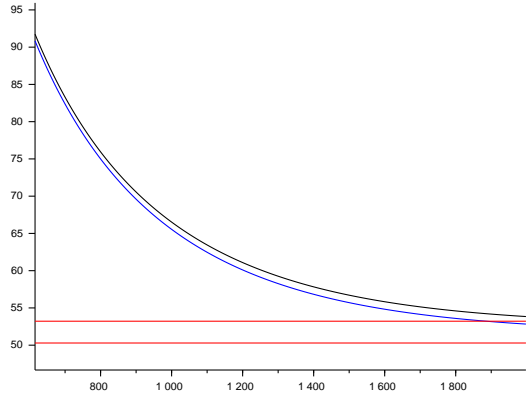}
\caption{Convergence of the Fourier transform price of the 6M$\times$1Y caplet with strike 1\% and a time step of $1/8$, in function of the number of discretization steps $n_{st}$. The integration is thus made on $[0,n_{st}/8]$. The parallel lines indicate the 95\% confidence interval obtained by MC.}
\label{FourierCv2}
\end{figure}

\section*{Conclusion}

The contribution of our paper is twofold. First, the purpose of this paper is to define a Wishart driven affine term structure model for interest rates model, in which the parameters and state variables of the model admit a clear interpretation in terms of the yield curve dynamics, and to provide an efficient numerical framework to implement the model. Other affine term structure models involving Wishart processes have been proposed for example by Bensusan~\cite{HB} or Gnoatto~\cite{Gnoatto}. A pitfall of general affine term structure model is to offer an abundant parametrization with few intuitions for the practitioner. Here, we believe that presenting the model as a perturbation of the standard LGM model is a good way to get a grip on it, to have a better understanding of the parameters and to have a starting point for the calibration procedure. Let us mention here that getting a reliable and stable calibration procedure of the model is beyond the scope of this paper. In particular, the choice of the dimensions $p$ and $d$ should be discussed on real data. Also, we have made the choice in this paper to present the model with constant (as opposite to time dependent) parameters: only the factors are meant to describe the state of the interest rate market. Thus, this version of the model has  a priori a limited flexility to calibrate to the swaption volatility cube compared to fully non-homogeneous term structure model with time dependent parameters such as the stochastic volatility forward Libor model of Piterbarg~\cite{citeulike:5804271}, and the stochastic volatility Cheyette model considered by Andreasen in~\cite{BackFuture}.  A full discussion on the calibration of our model as well as the comparison to other models is left for further research. 

The other contribution of the paper is to investigate different numerical methods for the model. We know that having efficient numerical methods is a prerequisite to use a model. Besides, our results can be interesting for other models based on Wishart dynamics. As the state variables dynamics is affine, their Fourier and Laplace transforms are tractable and can be obtained by solving Ordinary Differential Equations. Therefore Fourier transform pricing methods can be applied to price vanilla interest rates options in the model. However, the results of our numerical investigation suggest that standard Fourier based pricing methods suffer from numerical efficiency. This is due to the rather lengthy evaluation of the characteristic function together with a slow convergence rate of the Fourier transform discretization. A smarter way to evaluate the characteristic function and to solve the corresponding differential equation has to be investigated to make this method more attractive. As an alternative, we have developed a pricing method for vanilla interest rates options based on a perturbation of the infinitesimal generator of the state variables. This method provides a fast pricing tool for the products which would typically be used for model calibration. The method is particularly efficient for short expiries, but proves limitations for long dated options. Also, the expansion provides analytical expressions for the implied volatility of caplets and swaptions. This is important to confirm the intuitions on the role of the parameters and it can be used to initialize the calibration routine. Last, we propose a second order discretization scheme for the model, which is useful to run a Monte Carlo method. This scheme is easy to implement and very efficient in practice. Besides, it can be adapted easily to a wider range of financial models that use the same dependence structure between the vector~$Y$ and its instantaneous Wishart covariance matrix~$X$, such as the Wishart affine stochastic correlation model devloped by Da Fonseca et al.~\cite{DFGI,F}. Moreover, it is up to our knowledge the first second order discretization scheme that is able to handle this instantaneous covariance structure.

\appendix

\section{Explicit formulas of the price expansion}\label{ExpansionCalcul}

\subsection{Caplets price expansion}\label{CapletsExpansionCalcul}

We first write the expansion up to order~$1$ for $D$, and we get from~\eqref{LaplaceGenRiccati} that $D(t)=D_0(t) +\epsilon D_1(t) +O(\epsilon^2)$ with
$\dot{D}_0=D_0b+bD_0+\frac{1}{2}c^{\top}B B^{\top}c-\gamma$, $D_0(0)=0$ and $\dot{D}_1=D_1b+bD_1+\frac{1}{2}D_0 I^n_d \rho B^\top c+\frac{1}{2} c^\top B \rho^\top I^n_d D_0$, $D_1(0)=0$. We then obtain
\begin{eqnarray}
     \label{ExpSupportD0}
D_0(t)&=&e^{b^{\top} t}\left(\int_0^{t} e^{-b^{\top}s}\left(\frac{1}{2}c^{\top}B(s)B(s)^{\top}c-\gamma\right)e^{-bs} ds \right)e^{b t}\\
    \label{ExpSupportD1}
D_1(t)&=&\frac{1}{2}e^{b^{\top} t}\left(\int_0^{t} e^{-b^{\top}s}\left(c^{\top}B(s)\rho^{\top}I^n_dD_0(s)+D_0(s) I^n_d\rho B(s)^{\top}c\right)e^{-bs} ds \right)e^{b t}.
\end{eqnarray}
We recall that $X^0_{s}(x)$ is defined by~\eqref{defX0sx}.
The coefficients of formulas~\eqref{FormulaCapletExpOrder1} and~\eqref{FormulaCapletExpOrder2} are given by
\begin{align*}
c_1(t,T,\delta,x)&= \int_t^T  \Delta B^\top(s,T,\delta) c  X^0_{s-t}(x) \partial_x v (s,T,\delta)  I^n_d \rho ds, %\label{ExpCoeffC1}
\\
c_2(t,T,\delta,x)&= \int_t^T   \Delta B^{\top}(s,T,\delta) c X^0_{s-t}(x) \Delta D_0(s,T,\delta) I^n_d\rho +  B^\top(T+\delta-s) c X^0_{s-t}(x)  \partial_{x}v (s,T,\delta) I^n_d \rho ds,   %\label{ExpCoeffC2}
\end{align*}

\begin{align*}
e_1(t,T,\delta,x) =& \int_t^T  c_1(s,T,\delta,x) \Delta B^\top(s,T,\delta)   c X^0_{s-t}(x)  \partial_{x}v(s,T,\delta) I^n_d \rho ds, \\
e_2(t,T,\delta,x) =&  \int_t^T  c_1(s,T,\delta,x)  [ (\Delta B^{\top}(s,T,\delta) c X^0_{s-t}(x) \Delta D_0(s,T,\delta) I^n_d\rho) +B^\top(T+\delta-s) c X^0_{s-t}(x) \partial_{x}v(s,T,\delta) I^n_d \rho]  \\&+ c_2(s,T,\delta,x) \Delta B^\top(s,T,\delta)  c X^0_{s-t}(x) \partial_{x}v(s,T,\delta) I^n_d \rho ds, \\
e_3(t,T,\delta,x) =&     \int_t^T   c_2(s,T,\delta,x) [    (\Delta B^{\top}(s,T,\delta) c X^0_{s-t}(x) \Delta D_0(s,T,\delta) I^n_d\rho) + B^\top(T+\delta-s) c X^0_{s-t}(x) \partial_{x}v(s,T,\delta) I^n_d \rho ]ds ,      \\
e_4(t,T,\delta,x) =& \int_t^T     2 \Delta B^\top(s,T,\delta) c  X^0_{s-t}(x) \partial_{x}c_1 (s,T,\delta) I^n_d \rho ds ,\\
e_5(t,T,\delta,x) =& \int_t^T  2 B^\top(T+\delta-s)  c X^0_{s-t}(x) \partial_{x}c_1(s,T,\delta) I^n_d \rho    + 2 \Delta B^\top(s,T,\delta) c X^0_{s-t}(x)  \partial_{x}c_2(s,T,\delta) I^n_d \rho ds,\\
e_6(t,T,\delta,x) =&\int_t^T  2 B^\top(T+\delta-s) c X^0_{s-t}(x) \partial_{x}c_2(s,T,\delta) I^n_d \rho ds .\\
\end{align*}
and
\begin{align*}
d_1(t,T,\delta,x) =& \int_t^T \frac12 \Tr\left[ I^n_d   \partial_{x}v(s,T,\delta)  X^0_{s-t}(x) \partial_{x}v (s,T,\delta)\right] ds \\
d_2(t,T,\delta,x) =&  \int_t^T 2 \Tr\left[ \Delta D_0(s,T,\delta) X^0_{s-t}(x) \partial_{x}v(s,T,\delta) I^n_d \right]   ds \\
d_3(t,T,\delta,x) =&\int_t^T   \left(2 \Tr(\Delta D_0(s,T,\delta) I^n_d \Delta D_0(s,T,\delta) X^0_{s-t}(x) )+(\Delta B^{\top}(s,T,\delta) c X^0_{s-t}(x) \Delta D_1 (s,T,\delta) I^n_d\rho) \right)  \\
& +  \frac12 \Tr \left[ ((d-1)I^n_d+4X^0_{s-t}(x) D_0(T+\delta-s) I^n_d)   \partial_{x}v(s,T,\delta) \right] ds .
\end{align*}

\subsection{Swaption price expansion}\label{SwaptionExpansionCalcul}

We have
 $D^S=D_0^S+\epsilon D_1^S + o (\epsilon),$
  with  $$D_i^S(t) = \omega^0_0D_i(T-t)-\omega^m_0D_i(T+m\delta-t)-S_0(T,m,\delta)\sum_{k=1}^m\omega^k_0D_i(T+k\delta-t),\quad i=0,1,$$
 where the functions $D_0$ and $D_1$ are given by \eqref{ExpSupportD0} and \eqref{ExpSupportD1}.
The coefficients of formulas~\eqref{FormulaSwaptionExpOrder1} and~\eqref{FormulaSwaptionExpOrder2} are given by
\begin{align*}
c^S_1(t,T,x)=& \int_t^T   B^S(u)^\top c X^0_{u-t}(x) \partial_{x}v^S(u,T)  I^n_d \rho du, \\
c^S_2(t,T,x)=&\int_t^T  B^S(u)^{\top}c X^0_{u-t}(x) D_0^S(u) I^n_d\rho +  B^A(u)^\top c X^0_{u-t}(x) \partial_{x} v^S (u,T) I^n_d \rho du,
\end{align*}
\begin{align*}
d^S_1(t,T,x) =& \int_t^T  \frac{1}{2}  \Tr(I^n_d \partial_x v^S(u,T)  X^0_{u-t}(x) \partial_x v^S(u,T))   du, \\
d^S_2(t,T,x) =&  \int_t^T    2 \Tr(D^S_0(u) X^0_{u-t}(x)  \partial_x v^S(u,T) I^n_s)  du , \\
d^S_3(t,T,x) =&   \int_t^T [2 \Tr( D_0^S(u) I^n_d D_0^S(u)X^0_{u-t}(x) ) + (B^S(u))^{\top}c X^0_{u-t}(x)  D_1^S(u) I^n_d\rho] \\ &+ \Tr \left( [2X^0_{u-t}(x) D^A_0(u)I^n_d + \frac{1}{2}(d-1)I^n_d] \partial_x v^S(u,T)  \right)  du  ,
\end{align*}
\begin{align*}
e^S_1(t,T,x) =& \int_t^T  c^S_1(u,T,x) (B^S(u))^\top   c X^0_{u-t}(x)  \partial_{x}v^S(u,T) I^n_d \rho du, \\
e^S_2(t,T,x) =&  \int_t^T  c^S_1(u,T,x) \left[ B^S(u)^{\top}c X^0_{u-t}(x) D_0^S(u) I^n_d\rho +  B^A(u)^\top c X^0_{u-t}(x) \partial_{x}v^S(u,T)  I^n_d \rho \right]  \\
& + c^S_2(u,T,x)   B^S(u)^\top c X^0_{u-t}(x) \partial_{x}v^S(u,T)  I^n_d \rho   du , \\
e^S_3(t,T,x) =&     \int_t^T   c^S_2(u,T,x)  \left[ B^S(u)^{\top}c X^0_{u-t}(x) D_0^S(u) I^n_d \rho + B^A(u)^\top c X^0_{u-t}(x) \partial_{x}v^S(u,T)  I^n_d \rho \right] \\&+ 2    B^S(u)^\top c X^0_{u-t}(x) \partial_{x}c^S_1(u,T)  I^n_d \rho du ,      \\
e^S_4(t,T,x) =& \int_t^T  2  B^A(u)^\top c  X^0_{u-t}(x)  \partial_{x}c^S_1(u,T)  I^n_d \rho  + 2  B^S(u)^\top c  X^0_{u-t}(x)  \partial_{x}c^S_2(u,T)  I^n_d \rho     du ,\\
e^S_5(t,T,x) =& \int_t^T 2  B^A(u)^\top c X^0_{u-t}(x) \partial_{x}c^S_2(u,T)  I^n_d \rho   du.\\
\end{align*}

\section{Proof of Proposition~\ref{prop_ergo}}\label{ErgoExploSec}

We first recall the following useful result
\begin{equation}\label{trace_posit}
\forall x,y \in \posm, \ \Tr(xy)\ge 0,
\end{equation}
which comes easily from $\Tr(xy)=\Tr(\sqrt{x}y\sqrt{x})$ and $\sqrt{x}y\sqrt{x}\in \posm$.

For $x,y \in \symm$, we use the notation $x\le y$ if $y-x \in \posm$. By assumption, there is
$\mu>0$ such that $2 \mu I_d \le -(b+b^{\top})   $. We now apply Proposition~\ref{FourLapGeneralProp} with $\bar{\Lambda}=0$ and $\bar{\Gamma}=0$.
Since $\| \lambda(t)\| \le \|\Lambda \|$, there is a constant $h>0$ small enough such that for any $\Lambda \in \R^d$ satisfying $\|\Lambda \|< h$ we have
$$ \forall t\ge 0,  \mu I_d \le -[ b+ \frac{\epsilon}{2} I^n_d\rho \lambda^{\top} c +(b+\frac{\epsilon}{2} I^n_d\rho \lambda^{\top} c)^{\top}]  \text{ and }  \frac{1}{2}c^{\top} \lambda\lambda^{\top} c \le \frac{\mu^2}{8\epsilon^2}I_d$$
By choosing $\Upsilon=\frac{\mu}{4\epsilon^2}I_d$, we see that the condition~\eqref{cond2_ups} is satisfied since $\frac{\mu^2}{4\epsilon^2}I_d -\frac{\mu^2}{8\epsilon^2}I^n_d - \frac{1}{2}c^{\top} \lambda\lambda^{\top} c\in \posm $ for all $t\ge 0$. Thus, the conclusions of Proposition~\ref{FourLapGeneralProp} hold for any $\Lambda \in \R^d$ and $\Gamma \in \symm$ such that  $\|\Lambda \|< h$ and $\Gamma \le \frac{\mu}{4\epsilon^2}I_d $, and we have $\fctgm(t) \le \frac{\mu}{4\epsilon^2}I_d $ for any $t\ge 0$. We now want to prove that $\lambda(t)\underset{t \rightarrow +\infty}\rightarrow 0$, $\fctgm(t)\underset{t \rightarrow +\infty}\rightarrow 0$ and $\eta(t)$ converges when $t \rightarrow +\infty$. This will prove the convergence to the stationary law by L\'evy's theorem.

From~\eqref{LaplaceGenRiccati}, we have
$$\frac{1}{2} \frac{d}{dt}  \Tr(\fctgm^2) =2\epsilon^2 \Tr( \fctgm I^n_d\fctgm^2)+ \Tr(\fctgm^2 [ b+ \frac{\epsilon}{2} I^n_d\rho \lambda^{\top} c +(b+\frac{\epsilon}{2} I^n_d\rho \lambda^{\top} c)^{\top}] )+ \Tr( \fctgm \frac{1}{2}c^{\top} \lambda\lambda^{\top} c) .$$
By~\eqref{trace_posit}, we get
\begin{align*}
\frac{1}{2} \frac{d}{dt}  \Tr(\fctgm^2) & \le \frac{\mu}{2} \Tr( \fctgm I^n_d\fctgm) - \mu \Tr(\fctgm^2)+ \frac{\mu}{4\epsilon^2} \Tr(  \frac{1}{2}c^{\top} \lambda\lambda^{\top} c). %% \\
%% &\le - \frac{\mu}{2} \Tr(\fctgm^2) + \frac{\mu^3 d}{32\epsilon^3},
\end{align*}
Since  $\Tr( \fctgm I^n_d\fctgm)\le \Tr(\fctgm^2)$, we get by Gronwall's lemma
$$\frac{1}{2} \Tr(g(t)^2) \le \frac{1}{2} \Tr(\Gamma^2)e^{- \mu t}+  \frac{\mu}{4\epsilon^2}  \int_0^t\Tr\left(\left[ \frac{1}{2}c^\top  \lambda(s) \lambda^\top (s) c \right]^2 \right)e^{-\mu (t-s)} ds.$$
We now use that the entries of~$\lambda$ decay exponentially. Since $\int_0^te^{-\mu' s}e^{-\mu (t -s)}ds\underset{t \rightarrow +\infty}{=}O(e^{-\frac{\min(\mu,\mu')}{2}t})$ for $\mu,\mu'>0$, we get that there exists $C,\nu>0$ such that $\frac{1}{2} \Tr(g(t)^2) \le C e^{-\nu t}$. This gives that $\fctgm(t)\underset{t \rightarrow +\infty}\rightarrow 0$ and that $\eta(t)=\int_0^t \lambda^{\top}(s) \kappa\theta+\mathrm{Tr}\left(\fctgm(s)(\Omega+\epsilon^2(d-1)I^n_d)\right) ds$ converges. $\square$

\bibliographystyle{abbrv}
\bibliography{BiblioExpansion}

\begin{thebibliography}{10}

\bibitem{Alfonsi2}
A.~Ahdida and A.~Alfonsi.
\newblock Exact and high-order discretization schemes for {W}ishart processes
  and their affine extensions.
\newblock {\em Ann. Appl. Probab.}, 23(3):1025--1073, 2013.

\bibitem{Alfonsi1}
A.~Alfonsi.
\newblock High order discretization schemes for the {CIR} process: application
  to affine term structure and {H}eston models.
\newblock {\em Math. Comp.}, 79(269):209--237, 2010.

\bibitem{Piterbarg2010}
L.~Andersen and V.~Piterbarg.
\newblock {\em Interest Rate Modeling}, volume Issue 2.
\newblock Atlantic Financial Press, 2010.

\bibitem{BackFuture}
J.~Andreasen.
\newblock Back to the future.
\newblock {\em Risk Magazine}, 2005.

\bibitem{Benabid}
A.~Benabid, H.~Bensusan, and N.~El~Karoui.
\newblock {Wishart Stochastic Volatility: Asymptotic Smile and Numerical
  Framework}.
\newblock https://hal.archives-ouvertes.fr/hal-00458014, 2008.

\bibitem{HB}
H.~Bensussan.
\newblock Interest rate and longevity risks : dynamic modelling and
  applications to derivative products and life insurance.
\newblock {\em PhD Thesis, Ecole Polytechnique}, 2010.
\newblock https://tel.archives-ouvertes.fr/pastel-00563792/.

\bibitem{Bergomi}
L.~Bergomi and J.~Guyon.
\newblock Stochastic volatilities orderly smiles.
\newblock {\em Risk Magazine}, 2012.

\bibitem{BM}
D.~Brigo and F.~Mercurio.
\newblock {\em Interest rate models---theory and practice}.
\newblock Springer Finance. Springer-Verlag, Berlin, second edition, 2006.
\newblock With smile, inflation and credit.

\bibitem{CM}
P.~Carr and D.~Madan.
\newblock Option valuation using fast fourier transform.
\newblock {\em Journal of Computational Finance}, 1999.

\bibitem{DGE}
P.~Collin-Dufresne and R.~Goldstein.
\newblock Pricing swaptions within an affine framework.
\newblock {\em The Journal of Derivatives}, 2002.

\bibitem{CIR1}
J.~C. Cox, J.~E. Ingersoll, Jr., and S.~A. Ross.
\newblock A theory of the term structure of interest rates.
\newblock {\em Econometrica}, 53(2):385--407, 1985.

\bibitem{CFMT2011}
C.~Cuchiero, D.~Filipovi{\'c}, E.~Mayerhofer, and J.~Teichmann.
\newblock Affine processes on positive semidefinite matrices.
\newblock {\em Ann. Appl. Probab.}, 21(2):397--463, 2011.

\bibitem{DFGI}
J.~Da~Fonseca, M.~Grasselli, and F.~Ielpo.
\newblock Estimating the {W}ishart affine stochastic correlation model using
  the empirical characteristic function.
\newblock {\em Stud. Nonlinear Dyn. Econom.}, 18(3):253--289, 2014.

\bibitem{F}
J.~Da~Fonseca, M.~Grasselli, and C.~Tebaldi.
\newblock Option pricing when correlations are stochastic: an analytical
  framework.
\newblock {\em Review of Derivatives Research}, 2008.

\bibitem{DS1997}
Q.~Dai and K.~J. Singleton.
\newblock Specification analysis of affine term structure models.
\newblock {\em The Journal of Finance}, 55(5):1943--1978, 2000.

\bibitem{DBLP:journals/corr/cs-CE-0302034}
A.~D'Aspremont.
\newblock Interest rate model calibration using semidefinite programming.
\newblock {\em Applied Mathematical Finance}, 10(3):183--213, 2003.

\bibitem{DieciEirola}
L.~Dieci and T.~Eirola.
\newblock Positive definiteness in the numerical solution of {R}iccati
  differential equations.
\newblock {\em Numer. Math.}, 67(3):303--313, 1994.

\bibitem{DFS}
D.~Duffie, D.~Filipovi{\'c}, and W.~Schachermayer.
\newblock Affine processes and applications in finance.
\newblock {\em Ann. Appl. Probab.}, 13(3):984--1053, 2003.

\bibitem{DK1996}
D.~Duffie and R.~Kan.
\newblock A yield-factor model of interest rates.
\newblock {\em Mathematical Finance}, 6:379--406, 1996.

\bibitem{NEK1992}
N.~El~Karoui and V.~Lacoste.
\newblock Multifactor models of the term structure of interest rates.
\newblock 1992.

\bibitem{NEK1991}
N.~El~Karoui, C.~Lepage, R.~Myneni, N.~Roseau, and R.~Wiswanathan.
\newblock The valuation and hedging of contingent claims with markovian
  interest rates.
\newblock 1991.

\bibitem{Filipovic}
D.~Filipovi{\'c}.
\newblock {\em Term-structure models}.
\newblock Springer Finance. Springer-Verlag, Berlin, 2009.
\newblock A graduate course.

\bibitem{Fouque}
J.-P. Fouque, G.~Papanicolaou, and K.~R. Sircar.
\newblock {\em Derivatives in financial markets with stochastic volatility}.
\newblock Cambridge University Press, Cambridge, 2000.

\bibitem{NEK95}
H.~Geman, N.~El~Karoui, and J.-C. Rochet.
\newblock Changes of num\'eraire, changes of probability measure and option
  pricing.
\newblock {\em J. Appl. Probab.}, 32(2):443--458, 1995.

\bibitem{Gnoatto}
A.~Gnoatto.
\newblock The {W}ishart short rate model.
\newblock {\em Int. J. Theor. Appl. Finance}, 15(8):1250056, 24, 2012.

\bibitem{L}
R.~W. Lee.
\newblock Option pricing by transform methods: Extensions, unification, and
  error control.
\newblock {\em Journal of Computational Finance}, 2004.

\bibitem{Palidda}
E.~Palidda.
\newblock Managing interest rates derivatives with stochastic
  variance-covariance.
\newblock {\em PhD Thesis, Ecole des Ponts ParisTech}, 2015.
\newblock https://hal.archives-ouvertes.fr/tel-01217655v1.

\bibitem{citeulike:5804271}
V.~Piterbarg.
\newblock {A Stochastic Volatility Forward Libor Model with a Term Structure of
  Volatility Smiles}.
\newblock {\em Social Science Research Network Working Paper Series}, Nov.
  2003.

\bibitem{VP}
V.~Piterbarg.
\newblock Rates squared.
\newblock {\em Risk Magazine}, 2009.

\bibitem{Rydberg}
T.~H. Rydberg.
\newblock A note on the existence of unique equivalent martingale measures in a
  markovian setting.
\newblock {\em Finance and Stochastics}, 1(3):251--257, 1997.

\bibitem{citeulike:824660}
D.~F. Schrager and A.~A.~J. Pelsser.
\newblock Pricing swaptions and coupon bond options in affine term structure
  models.
\newblock {\em Math. Finance}, 16(4):673--694, 2006.

\bibitem{MAFI:MAFI427}
K.~J. Singleton and L.~Umantsev.
\newblock Pricing coupon-bond options and swaptions in affine term structure
  models.
\newblock {\em Math. Finance}, 12(4):427--446, 2002.

\bibitem{KT}
K.~Tanaka, T.~Yamada, and T.~Watanabe.
\newblock Applications of {G}ram-{C}harlier expansion and bond moments for
  pricing of interest rates and credit risk.
\newblock {\em Quant. Finance}, 10(6):645--662, 2010.

\bibitem{Vasicek}
O.~Vasicek.
\newblock An equilibrium characterization of the term structure.
\newblock {\em Journal of Financial Economics}, 5(2):177 -- 188, 1977.

\end{thebibliography}

\end{document}